\documentclass[a4paper,fullpage]{article}
\pdfoutput=1
\usepackage{xspace}
\usepackage{mathtools}
\usepackage{amssymb}
\usepackage{stmaryrd}
\usepackage{tabularx}
\usepackage{hyperref}
\usepackage{amsthm}
\theoremstyle{plain}

\newtheorem{lemma}{Lemma}
\newtheorem{definition}{Definition}
\newtheorem{theorem}{Theorem}
\newtheorem{corollary}{Corollary}
\newtheorem{remark}{Remark}
\newtheorem{proposition}[theorem]{Proposition}

\usepackage{geometry}
\usepackage{tikz}
\usepackage{pgfplots}
\usetikzlibrary{decorations.markings}
\usetikzlibrary{matrix,backgrounds,folding}
\pgfdeclarelayer{edgelayer}
\pgfdeclarelayer{nodelayer}
\pgfsetlayers{background,edgelayer,nodelayer,main}
\tikzstyle{every picture}=[baseline=-0.25em]
\tikzstyle{none}=[inner sep=0mm]
\tikzstyle{gn}=[shape=circle,fill=green, minimum width=.3cm, inner sep=0.5pt, font=\footnotesize, draw=black]
\tikzstyle{rn}=[shape=circle,fill=red,inner sep=0.5pt, minimum width=.3cm, font=\footnotesize, text=black, draw=black]
\tikzstyle{H box}=[rectangle,fill=yellow,draw=black,xscale=1,yscale=1,font=\small,inner sep=0.75pt,minimum width=0.15cm,minimum height=0.15cm]
\tikzstyle{sgn}=[shape=circle,fill=green,minimum width=.15cm,inner sep=0.5pt,font=\footnotesize, draw=black]
\tikzstyle{srn}=[shape=circle,fill=red,inner sep=0.5pt,minimum width=.15cm,text=black, draw=black]
\newcommand{\edgetick}{{\arrow[black,scale=0.7,very thick]{|}}}
\tikzstyle{tickedge}=[postaction=decorate,
  decoration={markings, mark=at position 0.5 with \edgetick}]

\allowdisplaybreaks

\newcommand{\interp}[1] {\left\llbracket #1 \right\rrbracket}
\newcommand{\interpspe}[1]{\raisebox{0.2em}{\scalebox{0.5}{$\left(\vphantom{\rule{1pt}{1.4em}}\right.\hspace{-0.8em}\raisebox{-0.2em}{\scalebox{2.1}{\rotatebox[origin=c]{90}{|}}}$}}#1\reflectbox{\raisebox{0.2em}{\scalebox{0.5}{$\left(\vphantom{\rule{1pt}{1.4em}}\right.\hspace{-0.8em}\raisebox{-0.2em}{\scalebox{2.1}{\rotatebox[origin=c]{90}{|}}}$}}}}
\newcommand{\fit}[1] {\resizebox{\columnwidth}{!}{#1}}

\newcommand{\frag}[1]{$\frac{\pi}{#1}$-fragment}
\newcommand{\abs}[1]{\begin{array}{|c|}#1\end{array}}
\renewcommand{\implies}{\quad\Rightarrow\quad}
\renewcommand{\mod}{\bmod}
\newcommand{\annoted}[3]{{\scriptstyle #1}\left\lbrace\mathrlap{\phantom{#3}}\right.\overbrace{#3}^{#2}}

\newcommand{\QQ}{\mathbb Q}
\newcommand{\CCC}{\mathbb C}

\newcommand{\invG}[1]{\left\llbracket #1 \right\rrbracket^{\begin{tikzpicture}[scale=0.5]
	\begin{pgfonlayer}{nodelayer}
		\node [style=sgn] (0) at (0, -0) {};
	\end{pgfonlayer}
\end{tikzpicture}}}
\newcommand{\invR}[1]{\left\llbracket #1 \right\rrbracket^{\begin{tikzpicture}[scale=0.5]
	\begin{pgfonlayer}{nodelayer}
		\node [style=srn] (0) at (0, -0) {};
	\end{pgfonlayer}
\end{tikzpicture}}}

\newcommand{\callrule}[2]{\hyperlink{r:#1}{\textnormal{(#2)}}\xspace}

\newcommand{\sone}{\callrule{rules}{S1}}
\newcommand{\stwo}{\callrule{rules}{S2}}
\newcommand{\sthree}{\callrule{rules}{S3}}
\newcommand{\iv}{\callrule{rules}{IV}}
\newcommand{\bo}{\callrule{rules}{B1}}
\newcommand{\bt}{\callrule{rules}{B2}}
\newcommand{\ko}{\callrule{rules}{K1}}
\newcommand{\kt}{\callrule{rules}{K2}}
\newcommand{\eu}{\callrule{rules}{EU}}
\newcommand{\h}{\callrule{rules}{H}}
\newcommand{\zo}{\callrule{rules}{ZO}}
\newcommand{\supp}{\callrule{rules}{SUP}}
\newcommand{\suppt}{\callrule{rules}{SUP$_2$}}
\newcommand{\suppn}{\callrule{SUP_n}{SUP$_n$}}
\newcommand{\suppc}[1]{\callrule{SUP_n}{SUP$_{#1}$}}
\newcommand{\e}{\callrule{E}{E}}
\newcommand{\hlaw}{\callrule{HL}{HL}}

\newcommand{\email}[1]{\texttt{\href{mailto:#1}{#1}}}

\title{ZX-Calculus: Cyclotomic Supplementarity and Incompleteness for Clifford+T quantum mechanics}

\author{Emmanuel Jeandel$^1$, Simon Perdrix$^1$, Renaud Vilmart$^1$, Quanlong Wang$^2$\\\\
{\small \begin{tabular}{>{\ttfamily}r>{\ttfamily}l}
1 & Universit\'e de Lorraine, CNRS, Inria, LORIA, F 54000 Nancy, France\\
2 & Dept.~of Comput.~Sci., Univ.~of Oxford, UK
\end{tabular}}}

\date{}

\begin{document}
\maketitle
\begin{center}
\begin{tabular}{rcl}
\email{emmanuel.jeandel@loria.fr}&$\qquad$&
\email{simon.perdrix@loria.fr}\\
\email{renaud.vilmart@loria.fr}&$\qquad$&
\email{quanlong.wang@cs.ox.ac.uk}
\end{tabular}
\end{center}
\begin{abstract}
The ZX-Calculus is a powerful graphical language for quantum mechanics and quantum information processing. The completeness of the language -- i.e.~the ability to derive any true equation -- is a crucial question. In the quest of a complete ZX-calculus, supplementarity has been recently proved to be necessary for quantum diagram reasoning (MFCS 2016). Roughly speaking, supplementarity consists in merging two subdiagrams when they are parameterized by antipodal angles. We introduce a generalised supplementarity -- called cyclotomic supplementarity -- which consists in merging $n$ subdiagrams at once, when the $n$ angles divide the circle into equal parts. We show that when $n$ is an odd prime number, the cyclotomic supplementarity cannot be derived, leading to a countable family of new axioms for diagrammatic quantum reasoning.

We exhibit another new simple axiom that cannot be derived from the existing rules of the ZX-Calculus, implying in particular the incompleteness of the language for the so-called Clifford+T quantum mechanics. We end up with a new axiomatisation of an extended ZX-Calculus, including an axiom schema for the cyclotomic supplementarity.
\end{abstract}

\section{Introduction}

The ZX-Calculus is a powerful diagrammatic language for reasoning in quantum mechanics introduced by Coecke and Duncan \cite{interacting}. Every diagram is composed of three kinds of vertices: red and green dots which are parametrised by an angle, and a yellow box; and each diagram represents a matrix thanks to the so-called standard interpretation. Moreover, any quantum transformation can be expressed using ZX-diagrams, meaning they are \emph{universal}. For instance, some particular states can be represented as evidenced by \cite{w-in-zx}. The language initially describes pure quantum state transformations, though some work has been made to adapt it to non pure evolutions \cite{Coecke2008,CP12} and measurement-based quantum computing \cite{mbqc}.

Unlike quantum circuits, the ZX-Calculus comes with a set of equalities between diagrams that preserve the matrix that is represented. Hence, using a succession of locally applied such equalities, one can prove that two diagrams represent the same matrix, for the language is \emph{sound} i.e.~all the equalities do indeed preserve the matrix.

The converse of soundness is called \textit{completeness}. Here, it amounts to being able to transform any diagram into another one, as long as both represent the same matrix. Hence, the concept of completeness is here totally defined by one particular interpretation, the \emph{standard interpretation}, unlike other definitions of completeness such as in \cite{hilbert-complete} in which it is related to a whole family of interpretations.

It has been proven that the ZX-Calculus is in general not complete \cite{incompleteness}. Yet, some restrictions have been proven to be complete. The \frag{2} -- the language restricted to angles that are multiples of $\frac{\pi}{2}$, which represents the \textit{stabiliser quantum mechanics} -- is complete \cite{pi_2-complete}, its pseudo-normal form using graph states introduced in the case of the ZX-Calculus in \cite{euler-decomp}. Moreover, this proof can be adapted to show the completeness of a ZX-like calculi used for graphically representing Spekken's toy model \cite{toy-model-graph,toy-model} or for graphically representing the real matrices \cite{Y-calculus}. The $\pi$-fragment -- representing the \textit{real stabiliser quantum mechanics} -- is also complete, with a slightly adapted set of rules \cite{pivoting}.

A fragment is \emph{approximately universal} when any quantum transformation can be approached with arbitrarily great precision using only the angles in the fragment. Sadly, the \frag{2} is not approximately universal, but the \frag{4} is \cite{clifford+t}. It is called the \emph{Clifford+T quantum mechanics}. Completeness for this fragment was an open question, one of the main ones in the fields of categorical quantum mechanics \cite{cqm} -- even though a partial answer has been given for the fragment composed of path diagrams involving angles multiple of $\frac{\pi}{4}$ \cite{pi_4-single-qubit}.

In this paper, we show that in the ZX-Calculus, the \frag{4} is not complete, showing that a scalar equality is derivable using matrices, but not diagrammatically. We propose to replace the ``inverse rule'' by this equality, and show that it can prove the former one as well as a third one: the ``zero rule''. Notice that this axiomatisation has been recently turned into complete axiomatisation of the ZX-calculus for this fragment \cite{complete}.

We also show that an infinite number of fragments are also incomplete, by proving that a generalised form of the ``supplementarity rule'' \cite{supplementarity} cannot be derived in them. Supplementarity, which has been proved to be necessary, consists in merging two subdiagrams when they are parameterized by antipodal angles. The generalised supplementarity – called cyclotomic supplementarity – consists in merging n subdiagrams at once, when the $n$ angles divide the circle into equal parts. We show that when n is an odd prime number, the cyclotomic supplementarity cannot be derived, leading to a countable family of new axioms for diagrammatic quantum reasoning.

Finally, we propose to add the new scalar equation, as well as the cyclotomic supplementarity to the set of rules, and to get rid of the now obsolete ``inverse'' and ``zero'' rules. We address the question of the incompleteness of the -- new -- general ZX-calculus with a modified version of the proof by Zamdzhiev and Schr\"oder de Witt \cite{incompleteness}, for theirs does not stand any more because of the introduction of the cyclotomic supplementarity.

We present the ZX-Calculus in section \ref{sec:zx}, prove the incompleteness of the \frag{4} and give a new scalar rule in section \ref{sec:pi_4}, and in section \ref{sec:generalised-supplementarity} we show how to generalise the supplementarity rule, discuss the way some are derivable from others, present the altered set of rules and give a new proof of incompleteness of the general ZX-Calculus.

\section{ZX-Calculus}
\label{sec:zx}

\subsection{Diagrams and standard interpretation}
A ZX-diagram $D:k\to l$ is an open graph with $k$ inputs and $l$ outputs -- read from top to bottom -- and is generated by:\\
\begin{center}
\bgroup
\def\arraystretch{2.5}
\begin{tabular}{|cc|cc|}
\hline
$R_Z^{(n,m)}(\alpha):n\to m$ & \begin{tikzpicture}
	\begin{pgfonlayer}{nodelayer}
		\node [style=gn] (0) at (0, -0) {$~\alpha~$};
		\node [style=none] (1) at (-0.5000001, 0.7499999) {};
		\node [style=none] (2) at (0, 0.4999999) {$\cdots$};
		\node [style=none] (3) at (0.5000001, 0.7499999) {};
		\node [style=none] (4) at (-0.5000001, -0.7499999) {};
		\node [style=none] (5) at (0, -0.4999999) {$\cdots$};
		\node [style=none] (6) at (0.5000001, -0.7499999) {};
		\node [style=none] (7) at (0, 0.75) {$n$};
		\node [style=none] (8) at (0, -0.7499998) {$m$};
	\end{pgfonlayer}
	\begin{pgfonlayer}{edgelayer}
		\draw [style=none, bend right, looseness=1.00] (1.center) to (0);
		\draw [style=none, bend right, looseness=1.00] (0) to (3.center);
		\draw [style=none, bend right, looseness=1.00] (0) to (4.center);
		\draw [style=none, bend left, looseness=1.00] (0) to (6.center);
	\end{pgfonlayer}
\end{tikzpicture} & $R_X^{(n,m)}(\alpha):n\to m$ & \begin{tikzpicture}
	\begin{pgfonlayer}{nodelayer}
		\node [style=rn] (0) at (0, -0) {$~\alpha~$};
		\node [style=none] (1) at (-0.5000001, 0.7499999) {};
		\node [style=none] (2) at (0, 0.4999999) {$\cdots$};
		\node [style=none] (3) at (0.5000001, 0.7499999) {};
		\node [style=none] (4) at (-0.5000001, -0.7499999) {};
		\node [style=none] (5) at (0, -0.5) {$\cdots$};
		\node [style=none] (6) at (0.5000001, -0.7499999) {};
		\node [style=none] (7) at (0, 0.75) {$n$};
		\node [style=none] (8) at (0, -0.7499999) {$m$};
		\node [style=none] (9) at (0, 1) {};
		\node [style=none] (10) at (0, -1) {};
	\end{pgfonlayer}
	\begin{pgfonlayer}{edgelayer}
		\draw [style=none, bend right, looseness=1.00] (1.center) to (0);
		\draw [style=none, bend right, looseness=1.00] (0) to (3.center);
		\draw [style=none, bend right, looseness=1.00] (0) to (4.center);
		\draw [style=none, bend left, looseness=1.00] (0) to (6.center);
	\end{pgfonlayer}
\end{tikzpicture}\\[4ex]\hline
$H:1\to 1$ & \begin{tikzpicture}
	\begin{pgfonlayer}{nodelayer}
		\node [style={H box}] (0) at (0.5, 0) {};
		\node [style=none] (1) at (0.5, 0.5) {};
		\node [style=none] (2) at (0.5, -0.5) {};
	\end{pgfonlayer}
	\begin{pgfonlayer}{edgelayer}
		\draw (1.center) to (0);
		\draw (2.center) to (0);
	\end{pgfonlayer}
\end{tikzpicture} & $e:0\to 0$ & \begin{tikzpicture}
	\begin{pgfonlayer}{nodelayer}
		\node [style=none] (0) at (-0.2499999, 0.2499999) {};
		\node [style=none] (1) at (-0.2499999, -0.2499999) {};
		\node [style=none] (2) at (0.2499999, 0.2499999) {};
		\node [style=none] (3) at (0.2499999, -0.2499999) {};
	\end{pgfonlayer}
	\begin{pgfonlayer}{edgelayer}
		\draw [style=dashed] (0.center) to (2.center);
		\draw [style=dashed] (2.center) to (3.center);
		\draw [style=dashed] (3.center) to (1.center);
		\draw [style=dashed] (0.center) to (1.center);
	\end{pgfonlayer}
\end{tikzpicture}\\\hline
$\mathbb{I}:1\to 1$ & \begin{tikzpicture}
	\begin{pgfonlayer}{nodelayer}
		\node [style=none] (0) at (0, 0.2499999) {};
		\node [style=none] (1) at (0, -0.2499999) {};
	\end{pgfonlayer}
	\begin{pgfonlayer}{edgelayer}
		\draw (0.center) to (1.center);
	\end{pgfonlayer}
\end{tikzpicture} & $\sigma:2\to 2$ & \begin{tikzpicture}
	\begin{pgfonlayer}{nodelayer}
		\node [style=none] (0) at (-0.2499999, 0.2499999) {};
		\node [style=none] (1) at (0.2499999, -0.2499999) {};
		\node [style=none] (2) at (0.2499999, 0.2499999) {};
		\node [style=none] (3) at (-0.2499999, -0.2499999) {};
	\end{pgfonlayer}
	\begin{pgfonlayer}{edgelayer}
		\draw [in=90, out=-90, looseness=1.00] (0.center) to (1.center);
		\draw [in=90, out=-90, looseness=1.00] (2.center) to (3.center);
	\end{pgfonlayer}
\end{tikzpicture}\\\hline
$\epsilon:2\to 0$ & \begin{tikzpicture}
	\begin{pgfonlayer}{nodelayer}
		\node [style=none] (0) at (-0.2499999, 0.2499999) {};
		\node [style=none] (1) at (0.2499999, 0.2499999) {};
	\end{pgfonlayer}
	\begin{pgfonlayer}{edgelayer}
		\draw [in=-90, out=-90, looseness=2.00] (0.center) to (1.center);
	\end{pgfonlayer}
\end{tikzpicture} & $\eta:0\to 2$ & \raisebox{0.6em}{\begin{tikzpicture}
	\begin{pgfonlayer}{nodelayer}
		\node [style=none] (0) at (0.2499999, -0.2499999) {};
		\node [style=none] (1) at (-0.2499999, -0.2499999) {};
	\end{pgfonlayer}
	\begin{pgfonlayer}{edgelayer}
		\draw [in=90, out=90, looseness=2.00] (1.center) to (0.center);
	\end{pgfonlayer}
\end{tikzpicture}}\\\hline
\end{tabular}
\egroup\\
where $n,m\in \mathbb{N}$ and $\alpha \in \mathbb{R}$
\end{center}

and the two compositions:
\begin{itemize}
\item Spacial Composition: for any $D_1:a\to b$ and $D_2:c\to d$, $D_1\otimes D_2:a+c\to b+d$ consists in placing $D_1$ and $D_2$ side by side, $D_2$ on the right of $D_1$.
\item Sequential Composition: for any $D_1:a\to b$ and $D_2:b\to c$, $D_2\circ D_1:a\to c$ consists in placing $D_1$ on the top of $D_2$, connecting the outputs of $D_1$ to the inputs of $D_2$.
\end{itemize}


The standard interpretation of the stabiliser ZX-diagrams associates to any diagram $D:n\to m$ a linear map $\interp{D}:\mathbb{C}^{2^n}\to\mathbb{C}^{2^m}$ inductively defined as follows:
$$ \interp{D_1\otimes D_2}:=\interp{D_1}\otimes\interp{D_2} \qquad 
\interp{D_2\circ D_1}:=\interp{D_2}\circ\interp{D_1} $$
$$
\interp{~\begin{tikzpicture}
	\begin{pgfonlayer}{nodelayer}
		\node [style=none] (0) at (-0.2499999, 0.2499999) {};
		\node [style=none] (1) at (-0.2499999, -0.2499999) {};
		\node [style=none] (2) at (0.2499999, 0.2499999) {};
		\node [style=none] (3) at (0.2499999, -0.2499999) {};
	\end{pgfonlayer}
	\begin{pgfonlayer}{edgelayer}
		\draw [style=dashed] (0.center) to (2.center);
		\draw [style=dashed] (2.center) to (3.center);
		\draw [style=dashed] (3.center) to (1.center);
		\draw [style=dashed] (0.center) to (1.center);
	\end{pgfonlayer}
\end{tikzpicture}~}:=\begin{pmatrix}1\end{pmatrix}
\qquad\interp{~~\begin{tikzpicture}
	\begin{pgfonlayer}{nodelayer}
		\node [style=none] (0) at (0, 0.2499999) {};
		\node [style=none] (1) at (0, -0.2499999) {};
	\end{pgfonlayer}
	\begin{pgfonlayer}{edgelayer}
		\draw (0.center) to (1.center);
	\end{pgfonlayer}
\end{tikzpicture}~~}:= \begin{pmatrix}
1 & 0 \\ 0 & 1\end{pmatrix}\qquad
\interp{~\begin{tikzpicture}
	\begin{pgfonlayer}{nodelayer}
		\node [style={H box}] (0) at (0, 0) {};
		\node [style=none] (1) at (0, 0.5) {};
		\node [style=none] (2) at (0, -0.5) {};
	\end{pgfonlayer}
	\begin{pgfonlayer}{edgelayer}
		\draw (2.center) to (1.center);
	\end{pgfonlayer}
\end{tikzpicture}
~}:= \frac{1}{\sqrt{2}}\begin{pmatrix}1 & 1\\1 & -1\end{pmatrix}$$
$$
\interp{\begin{tikzpicture}
	\begin{pgfonlayer}{nodelayer}
		\node [style=none] (0) at (-0.2499999, 0.2499999) {};
		\node [style=none] (1) at (0.2499999, -0.2499999) {};
		\node [style=none] (2) at (0.2499999, 0.2499999) {};
		\node [style=none] (3) at (-0.2499999, -0.2499999) {};
	\end{pgfonlayer}
	\begin{pgfonlayer}{edgelayer}
		\draw [in=90, out=-90, looseness=1.00] (0.center) to (1.center);
		\draw [in=90, out=-90, looseness=1.00] (2.center) to (3.center);
	\end{pgfonlayer}
\end{tikzpicture}}:= \begin{pmatrix}
1&0&0&0\\
0&0&1&0\\
0&1&0&0\\
0&0&0&1
\end{pmatrix} \qquad
\interp{\raisebox{-0.25em}{$\begin{tikzpicture}
	\begin{pgfonlayer}{nodelayer}
		\node [style=none] (0) at (-0.2499999, 0.2499999) {};
		\node [style=none] (1) at (0.2499999, 0.2499999) {};
	\end{pgfonlayer}
	\begin{pgfonlayer}{edgelayer}
		\draw [in=-90, out=-90, looseness=2.00] (0.center) to (1.center);
	\end{pgfonlayer}
\end{tikzpicture}$}}:= \begin{pmatrix}
1&0&0&1
\end{pmatrix} \qquad
\interp{\raisebox{0.4em}{$\begin{tikzpicture}
	\begin{pgfonlayer}{nodelayer}
		\node [style=none] (0) at (0.2499999, -0.2499999) {};
		\node [style=none] (1) at (-0.2499999, -0.2499999) {};
	\end{pgfonlayer}
	\begin{pgfonlayer}{edgelayer}
		\draw [in=90, out=90, looseness=2.00] (1.center) to (0.center);
	\end{pgfonlayer}
\end{tikzpicture}$}}:= \begin{pmatrix}
1\\0\\0\\1
\end{pmatrix}$$

$$
\interp{\begin{tikzpicture}
	\begin{pgfonlayer}{nodelayer}
		\node [style=gn] (0) at (0, -0) {$\alpha$};
	\end{pgfonlayer}
\end{tikzpicture}}:=\begin{pmatrix}1+e^{i\alpha}\end{pmatrix} \qquad
\interp{\begin{tikzpicture}
	\begin{pgfonlayer}{nodelayer}
		\node [style=gn] (0) at (0, -0) {$~\alpha~$};
		\node [style=none] (1) at (-0.5000001, 0.7499999) {};
		\node [style=none] (2) at (0, 0.4999999) {$\cdots$};
		\node [style=none] (3) at (0.5000001, 0.7499999) {};
		\node [style=none] (4) at (-0.5000001, -0.7499999) {};
		\node [style=none] (5) at (0, -0.4999999) {$\cdots$};
		\node [style=none] (6) at (0.5000001, -0.7499999) {};
		\node [style=none] (7) at (0, 0.75) {$n$};
		\node [style=none] (8) at (0, -0.7499998) {$m$};
	\end{pgfonlayer}
	\begin{pgfonlayer}{edgelayer}
		\draw [style=none, bend right, looseness=1.00] (1.center) to (0);
		\draw [style=none, bend right, looseness=1.00] (0) to (3.center);
		\draw [style=none, bend right, looseness=1.00] (0) to (4.center);
		\draw [style=none, bend left, looseness=1.00] (0) to (6.center);
	\end{pgfonlayer}
\end{tikzpicture}}:=
\annoted{2^m}{2^n}{\begin{pmatrix}
  1 & 0 & \cdots & 0 & 0 \\
  0 & 0 & \cdots & 0 & 0 \\
  \vdots & \vdots & \ddots & \vdots & \vdots \\
  0 & 0 & \cdots & 0 & 0 \\
  0 & 0 & \cdots & 0 & e^{i\alpha}
 \end{pmatrix}}
~~\begin{pmatrix}n+m>0\end{pmatrix} 
$$

For any $n,m\geq 0$ and $\alpha\in\mathbb{R}$, $\interp{\begin{tikzpicture}
	\begin{pgfonlayer}{nodelayer}
		\node [style=rn] (0) at (0, -0) {$~\alpha~$};
		\node [style=none] (1) at (-0.5000001, 0.7499999) {};
		\node [style=none] (2) at (0, 0.4999999) {$\cdots$};
		\node [style=none] (3) at (0.5000001, 0.7499999) {};
		\node [style=none] (4) at (-0.5000001, -0.7499999) {};
		\node [style=none] (5) at (0, -0.5) {$\cdots$};
		\node [style=none] (6) at (0.5000001, -0.7499999) {};
		\node [style=none] (7) at (0, 0.75) {$n$};
		\node [style=none] (8) at (0, -0.7499999) {$m$};
		\node [style=none] (9) at (0, 1) {};
		\node [style=none] (10) at (0, -1) {};
	\end{pgfonlayer}
	\begin{pgfonlayer}{edgelayer}
		\draw [style=none, bend right, looseness=1.00] (1.center) to (0);
		\draw [style=none, bend right, looseness=1.00] (0) to (3.center);
		\draw [style=none, bend right, looseness=1.00] (0) to (4.center);
		\draw [style=none, bend left, looseness=1.00] (0) to (6.center);
	\end{pgfonlayer}
\end{tikzpicture}}=\interp{~\begin{tikzpicture}
	\begin{pgfonlayer}{nodelayer}
		\node [style={H box}] (0) at (0, 0) {};
		\node [style=none] (1) at (0, 0.5) {};
		\node [style=none] (2) at (0, -0.5) {};
	\end{pgfonlayer}
	\begin{pgfonlayer}{edgelayer}
		\draw (2.center) to (1.center);
	\end{pgfonlayer}
\end{tikzpicture}
~}^{\otimes m}\circ \interp{\begin{tikzpicture}
	\begin{pgfonlayer}{nodelayer}
		\node [style=gn] (0) at (0, -0) {$~\alpha~$};
		\node [style=none] (1) at (-0.5000001, 0.7499999) {};
		\node [style=none] (2) at (0, 0.4999999) {$\cdots$};
		\node [style=none] (3) at (0.5000001, 0.7499999) {};
		\node [style=none] (4) at (-0.5000001, -0.7499999) {};
		\node [style=none] (5) at (0, -0.4999999) {$\cdots$};
		\node [style=none] (6) at (0.5000001, -0.7499999) {};
		\node [style=none] (7) at (0, 0.75) {$n$};
		\node [style=none] (8) at (0, -0.7499998) {$m$};
	\end{pgfonlayer}
	\begin{pgfonlayer}{edgelayer}
		\draw [style=none, bend right, looseness=1.00] (1.center) to (0);
		\draw [style=none, bend right, looseness=1.00] (0) to (3.center);
		\draw [style=none, bend right, looseness=1.00] (0) to (4.center);
		\draw [style=none, bend left, looseness=1.00] (0) to (6.center);
	\end{pgfonlayer}
\end{tikzpicture}}\circ \interp{~\begin{tikzpicture}
	\begin{pgfonlayer}{nodelayer}
		\node [style={H box}] (0) at (0, 0) {};
		\node [style=none] (1) at (0, 0.5) {};
		\node [style=none] (2) at (0, -0.5) {};
	\end{pgfonlayer}
	\begin{pgfonlayer}{edgelayer}
		\draw (2.center) to (1.center);
	\end{pgfonlayer}
\end{tikzpicture}
~}^{\otimes n}$ \\
$\left(\text{where }M^{\otimes 0}=\begin{pmatrix}1\end{pmatrix}\text{ and }M^{\otimes k}=M\otimes M^{\otimes k-1}\text{ for any }k\in \mathbb{N}^*:=\mathbb{N}\setminus\{0\}\right)$.\\

\noindent
To simplify, the red and green nodes will be represented empty when holding a 0 angle:
\[ \begin{tikzpicture}
	\begin{pgfonlayer}{nodelayer}
		\node [style=gn] (0) at (-0.9999999, -0) {};
		\node [style=none] (1) at (-1.5, 0.7499999) {};
		\node [style=none] (2) at (-0.9999999, 0.5) {$\cdots$};
		\node [style=none] (3) at (-0.5000002, 0.7499999) {};
		\node [style=none] (4) at (-1.5, -0.7500001) {};
		\node [style=none] (5) at (-0.9999999, -0.5) {$\cdots$};
		\node [style=none] (6) at (-0.5000002, -0.7500001) {};
		\node [style=gn] (7) at (0.9999999, -0) {0};
		\node [style=none] (8) at (0.9999999, -0.5) {$\cdots$};
		\node [style=none] (9) at (0.9999999, 0.5) {$\cdots$};
		\node [style=none] (10) at (0.5000002, -0.7500001) {};
		\node [style=none] (11) at (1.5, -0.7500001) {};
		\node [style=none] (12) at (0.5000002, 0.7499999) {};
		\node [style=none] (13) at (1.5, 0.7499999) {};
		\node [style=none] (14) at (0, -0) {:=};
	\end{pgfonlayer}
	\begin{pgfonlayer}{edgelayer}
		\draw [style=none, bend right, looseness=1.00] (1.center) to (0);
		\draw [style=none, bend right, looseness=1.00] (0) to (3.center);
		\draw [style=none, bend right, looseness=1.00] (0) to (4.center);
		\draw [style=none, bend left, looseness=1.00] (0) to (6.center);
		\draw [style=none, bend right, looseness=1.00] (12.center) to (7);
		\draw [style=none, bend right, looseness=1.00] (7) to (13.center);
		\draw [style=none, bend right, looseness=1.00] (7) to (10.center);
		\draw [style=none, bend left, looseness=1.00] (7) to (11.center);
	\end{pgfonlayer}
\end{tikzpicture} \qquad\text{and}\qquad \begin{tikzpicture}
	\begin{pgfonlayer}{nodelayer}
		\node [style=rn] (0) at (-0.9999999, -0) {};
		\node [style=none] (1) at (-1.5, 0.7499999) {};
		\node [style=none] (2) at (-0.9999999, 0.5) {$\cdots$};
		\node [style=none] (3) at (-0.5000002, 0.7499999) {};
		\node [style=none] (4) at (-1.5, -0.7500001) {};
		\node [style=none] (5) at (-0.9999999, -0.5) {$\cdots$};
		\node [style=none] (6) at (-0.5000002, -0.7500001) {};
		\node [style=rn] (7) at (0.9999999, -0) {0};
		\node [style=none] (8) at (0.9999999, -0.5) {$\cdots$};
		\node [style=none] (9) at (0.9999999, 0.5) {$\cdots$};
		\node [style=none] (10) at (0.5000002, -0.7500001) {};
		\node [style=none] (11) at (1.5, -0.7500001) {};
		\node [style=none] (12) at (0.5000002, 0.7499999) {};
		\node [style=none] (13) at (1.5, 0.7499999) {};
		\node [style=none] (14) at (0, -0) {:=};
	\end{pgfonlayer}
	\begin{pgfonlayer}{edgelayer}
		\draw [style=none, bend right, looseness=1.00] (1.center) to (0);
		\draw [style=none, bend right, looseness=1.00] (0) to (3.center);
		\draw [style=none, bend right, looseness=1.00] (0) to (4.center);
		\draw [style=none, bend left, looseness=1.00] (0) to (6.center);
		\draw [style=none, bend right, looseness=1.00] (12.center) to (7);
		\draw [style=none, bend right, looseness=1.00] (7) to (13.center);
		\draw [style=none, bend right, looseness=1.00] (7) to (10.center);
		\draw [style=none, bend left, looseness=1.00] (7) to (11.center);
	\end{pgfonlayer}
\end{tikzpicture} \]
Also in order to make the diagrams a little less heavy, when $n$ copies of the same subdiagram occur, we will use the notation $(.)^{\otimes n}$ as defined above.\\

With the general calculus -- with angles being in $\mathbb{R}$ -- we can represent any matrix of size a power of $2$ i.e.~ZX-Diagrams are universal:
\[\forall A\in \mathbb{C}^{2^n}\times\mathbb{C}^{2^m},~~\exists D,~~ \interp{D}=A\]
This requires dealing with an uncountable set of angles, so it is generally preferred to work with \textit{approximate} universality -- the ability to approximate any linear map with arbitrary accuracy -- in which only a finite set of angles is involved. The \frag{4} -- ZX-diagrams where all angles are multiples of $\frac{\pi}{4}$ -- is one such approximately universal fragment, whereas the \frag{2} is not.

\subsection{Calculus}
The diagrammatic representation of a matrix is not unique in the ZX-Calculus. Hence, a set of equalities has been proposed to axiomatise the language. This set is summed up in Figure \ref{fig:ZX_rules}.
\begin{figure}[!hbt]
 \centering
 \hypertarget{r:rules}{}
  \begin{tabular}{|ccccc|}
   \hline
   &&&& \\
   \begin{tikzpicture}[font={\footnotesize}]
	\begin{pgfonlayer}{nodelayer}
		\node [style=none] (0) at (-1.25, -0) {\rotatebox[origin=c]{63.43}{$~\cdots~$}};
		\node [style=none] (1) at (0.25, -0) {$=$};
		\node [style=gn] (2) at (1.5, 0) { \footnotesize$\alpha{+}\beta$};
		\node [style=gn, minimum width={0.5 cm}] (3) at (-0.7500001, -0.2500001) {\footnotesize $\beta$};
		\node [style=none] (4) at (-1.75, -0.5) {$~\cdots~$};
		\node [style=none] (5) at (2, -0.75) {};
		\node [style=none] (6) at (-1, -0.75) {};
		\node [style=none] (7) at (1.5, -0.75) {$~\cdots~$};
		\node [style=none] (8) at (-0.5, -0.75) {};
		\node [style=none] (9) at (1, -0.75) {};
		\node [style=none] (10) at (-2, -0.5) {};
		\node [style=none] (11) at (-1.5, -0.5) {};
		\node [style=none] (12) at (-0.75, -0.75) {$~\cdots~$};
		\node [style=none] (13) at (1.5, 0.75) {$~\cdots~$};
		\node [style=none] (14) at (1, 0.75) {};
		\node [style=none] (15) at (-2, 0.75) {};
		\node [style=none] (16) at (-0.5, 0.5) {};
		\node [style=none] (17) at (-1.5, 0.75) {};
		\node [style=none] (18) at (2, 0.75) {};
		\node [style=gn, minimum width={0.5 cm}] (19) at (-1.75, 0.25) {\footnotesize$\alpha$};
		\node [style=none] (20) at (-0.75, 0.5) {$~\cdots~$};
		\node [style=none] (21) at (-1, 0.5) {};
		\node [style=none] (22) at (-1.75, 0.75) {$~\cdots~$};
	\end{pgfonlayer}
	\begin{pgfonlayer}{edgelayer}
		\draw (3) to (16.center);
		\draw (3) to (6.center);
		\draw (3) to (8.center);
		\draw (19) to (10.center);
		\draw (19) to (11.center);
		\draw [bend right, looseness=1.00] (19) to (3);
		\draw [bend left, looseness=1.00] (19) to (3);
		\draw (14.center) to (2);
		\draw (2) to (9.center);
		\draw (5.center) to (2);
		\draw (2) to (18.center);
		\draw (19) to (15.center);
		\draw (19) to (17.center);
		\draw (3) to (21.center);
	\end{pgfonlayer}
\end{tikzpicture}&\hypertarget{r:S1}{(S1)} &$\qquad$& \begin{tikzpicture}
	\begin{pgfonlayer}{nodelayer}
		\node [style=gn] (0) at (-0.7499998, -0) {};
		\node [style=none] (1) at (0, -0) {=};
		\node [style=none] (2) at (-0.7499998, 1) {};
		\node [style=none] (3) at (-0.7499998, -0.9999999) {};
		\node [style=none] (4) at (0.7499998, 1) {};
		\node [style=none] (5) at (0.7499998, -0.9999999) {};
	\end{pgfonlayer}
	\begin{pgfonlayer}{edgelayer}
		\draw (2) to (0);
		\draw (0) to (3);
		\draw (4) to (5);
	\end{pgfonlayer}
\end{tikzpicture}&\hypertarget{r:S2}{(S2)}\\
   &&&& \\
   \begin{tikzpicture}
	\begin{pgfonlayer}{nodelayer}
		\node [style=none] (0) at (0.7500001, -0.25) {};
		\node [style=none] (1) at (0, -0) {$=$};
		\node [style=gn] (2) at (1.25, 0.25) {};
		\node [style=none] (3) at (-0.7500001, -0.25) {};
		\node [style=none] (4) at (1.75, -0.25) {};
		\node [style=none] (5) at (-1.75, -0.25) {};
	\end{pgfonlayer}
	\begin{pgfonlayer}{edgelayer}
		\draw [in=90, out=90, looseness=1.75] (5.center) to (3.center);
		\draw [in=90, out=90, looseness=1.75] (0.center) to (4.center);
	\end{pgfonlayer}
\end{tikzpicture}&\hypertarget{r:S3}{(S3)} && \begin{tikzpicture}
	\begin{pgfonlayer}{nodelayer}
		\node [style=rn] (0) at (-0.5000001, 0.25) {};
		\node [style=gn] (1) at (-0.5000002, -0.2500001) {};
		\node [style=gn] (2) at (-0.9999998, -0.2500001) {};
		\node [style=rn] (3) at (-0.9999998, 0.2500001) {};
		\node [style=none] (4) at (0.9999998, 0.2500001) {};
		\node [style=none] (5) at (0.4999996, 0.2500001) {};
		\node [style=none] (6) at (0.9999998, -0.2500001) {};
		\node [style=none] (7) at (0.4999996, -0.2500001) {};
		\node [style=none] (8) at (0, -0) {$=$};
	\end{pgfonlayer}
	\begin{pgfonlayer}{edgelayer}
		\draw [bend left=45, looseness=1.00] (0) to (1);
		\draw (3) to (2);
		\draw [bend right=45, looseness=1.00] (0) to (1);
		\draw (0) to (1);
		\draw [color=gray, dashed] (5.center) to (7.center);
		\draw [color=gray, dashed] (7.center) to (6.center);
		\draw [color=gray, dashed] (6.center) to (4.center);
		\draw [color=gray, dashed] (4.center) to (5.center);
	\end{pgfonlayer}
\end{tikzpicture}&\hypertarget{r:IV}{(IV)}\\
   &&&& \\
   \begin{tikzpicture}
	\begin{pgfonlayer}{nodelayer}
		\node [style=gn] (0) at (0.75, 0) {};
		\node [style=none] (1) at (2.25, -0.25) {};
		\node [style=none] (2) at (0.5, -0.5) {};
		\node [style=rn] (3) at (2.25, 0.25) {};
		\node [style=none] (4) at (1, -0.5) {};
		\node [style=rn] (5) at (0.75, 0.5) {};
		\node [style=rn] (6) at (2.75, 0.25) {};
		\node [style=none] (7) at (2.75, -0.25) {};
		\node [style=none] (8) at (1.5, 0) {$=$};
		\node [style=rn] (9) at (0, 0.25) {};
		\node [style=gn] (10) at (0, -0.25) {};
	\end{pgfonlayer}
	\begin{pgfonlayer}{edgelayer}
		\draw [style=none] (5) to (0);
		\draw[bend right=23]  [style=none] (0) to (2.center);
		\draw[bend left=23]  [style=none] (0) to (4.center);
		\draw [style=none] (3) to (1.center);
		\draw [style=none] (6) to (7.center);
		\draw (9) to (10);
	\end{pgfonlayer}
\end{tikzpicture}&\hypertarget{r:B1}{(B1)} && \begin{tikzpicture}
	\begin{pgfonlayer}{nodelayer}
		\node [style=none] (0) at (3.75, 0.75) {};
		\node [style=rn] (1) at (0.5, -0.25) {};
		\node [style=none] (2) at (3.25, -0.75) {};
		\node [style=none] (3) at (1.25, 1) {};
		\node [style=none] (4) at (3.25, 0.75) {};
		\node [style=none] (5) at (0.5, -0.75) {};
		\node [style=none] (6) at (0.5, 1) {};
		\node [style=gn] (7) at (1.25, 0.5) {};
		\node [style=none] (8) at (2.25, 0) {$=$};
		\node [style=rn] (9) at (1.25, -0.25) {};
		\node [style=gn] (10) at (3.5, -0.25) {};
		\node [style=gn] (11) at (0.5, 0.5) {};
		\node [style=none] (12) at (1.25, -0.75) {};
		\node [style=none] (13) at (3.75, -0.75) {};
		\node [style=rn] (14) at (3.5, 0.25) {};
		\node [style=rn] (15) at (0, 0.25) {};
		\node [style=gn] (16) at (0, -0.25) {};
	\end{pgfonlayer}
	\begin{pgfonlayer}{edgelayer}
		\draw [style=none] (12.center) to (9);
		\draw [style=none] (5.center) to (1);
		\draw [style=none] (7) to (3.center);
		\draw [style=none, bend right=23, looseness=1.00] (9) to (7);
		\draw [style=none] (11) to (6.center);
		\draw [style=none, bend left=23, looseness=1.00] (1) to (11);
		\draw [style=none,bend right=23] (13.center) to (10);
		\draw [style=none] (10) to (14);
		\draw [style=none,bend left=23] (14) to (4.center);
		\draw [style=none,bend right=23] (14) to (0.center);
		\draw[bend right=23] (10) to (2.center);
		\draw (11) to (9);
		\draw (7) to (1);
		\draw (15) to (16);
	\end{pgfonlayer}
\end{tikzpicture}&\hypertarget{r:B2}{(B2)}\\
   &&&& \\
   \begin{tikzpicture}
	\begin{pgfonlayer}{nodelayer}
		\node [style=rn] (0) at (-3.25, -0) {};
		\node [style=none] (1) at (-1.5, -0.5) {};
		\node [style=none] (2) at (-2.25, 0.25) {$=$};
		\node [style=gn, minimum width={0.4 cm}] (3) at (-1.5, -0) {\footnotesize$\pi$};
		\node [style=gn, minimum width={0.4 cm}] (4) at (-1, -0) {\footnotesize$\pi$};
		\node [style=none] (5) at (-1, -0.5) {};
		\node [style=gn, minimum width={0.4 cm}] (6) at (-3.25, 0.5) {\footnotesize$\pi$};
		\node [style=none] (7) at (-3.25, 1) {};
		\node [style=none] (8) at (-1.25, 1) {};
		\node [style=none] (9) at (-3, -0.5) {};
		\node [style=rn] (10) at (-1.25, 0.5) {};
		\node [style=none] (11) at (-3.5, -0.5) {};
		\node [style=none] (12) at (-3.25, -0.5) {$~\cdots~$};
		\node [style=none] (13) at (-1.25, -0.5) {$~\cdots~$};
	\end{pgfonlayer}
	\begin{pgfonlayer}{edgelayer}
		\draw [style=none] (0) to (11.center);
		\draw [style=none] (0) to (9.center);
		\draw [style=none] (8.center) to (10);
		\draw [style=none] (7.center) to (6);
		\draw [style=none] (6) to (0);
		\draw [style=none] (10) to (3);
		\draw [style=none] (3) to (1.center);
		\draw [style=none] (10) to (4);
		\draw [style=none] (4) to (5.center);
	\end{pgfonlayer}
\end{tikzpicture}&\hypertarget{r:K1}{(K1)} && \begin{tikzpicture}
	\begin{pgfonlayer}{nodelayer}
		\node [style=none] (0) at (0, -0) {$=$};
		\node [style=none] (1) at (-0.7499998, 0.9999999) {};
		\node [style=none] (2) at (-0.7499998, -0.7499998) {};
		\node [style=none] (3) at (1.5, 0.9999999) {};
		\node [style=gn] (4) at (-0.7499998, -0.2500001) {$~\pi~$};
		\node [style=none] (5) at (1.5, -0.7499998) {};
		\node [style=rn] (6) at (-0.7499998, 0.5) {$~\alpha~$};
		\node [style=rn] (7) at (1.5, -0.2500001) {$-\alpha$};
		\node [style=gn] (8) at (1.5, 0.5) {$~\pi~$};
		\node [style=rn] (9) at (-1.5, 0.2500001) {};
		\node [style=gn] (10) at (-1.5, -0.2500001) {};
		\node [style=rn] (11) at (0.7499998, 0.5) {$~\alpha~$};
		\node [style=gn] (12) at (0.7499998, -0.2500001) {$~\pi~$};
	\end{pgfonlayer}
	\begin{pgfonlayer}{edgelayer}
		\draw (3.center) to (8);
		\draw (8) to (7);
		\draw (7) to (5.center);
		\draw (4) to (2.center);
		\draw (1.center) to (6);
		\draw (6) to (4);
		\draw (9) to (10);
		\draw (12) to (11);
	\end{pgfonlayer}
\end{tikzpicture}&\hypertarget{r:K2}{(K2)}\\
   &&&& \\
   \begin{tikzpicture}
	\begin{pgfonlayer}{nodelayer}
		\node [style=gn] (0) at (0.5, -0.55) {$~\frac{\pi}{2}~$};
		\node [style=rn] (1) at (0.5, -0) {};
		\node [style=gn] (2) at (0.5, 0.55) {$~\frac{\pi}{2}~$};
		\node [style=none] (3) at (0.5, 0.9999999) {};
		\node [style=gn] (4) at (1.25, 0.5) {$\frac{-\pi}{2}$};
		\node [style=none] (5) at (0.5, -0.9999999) {};
		\node [style=none] (6) at (-0.9999999, 0.9999999) {};
		\node [style=none] (7) at (-0.9999999, -0.9999999) {};
		\node [style=none] (8) at (-0.2500001, -0) {$=$};
		\node [style={{H box}}] (9) at (-0.9999999, -0) {};
	\end{pgfonlayer}
	\begin{pgfonlayer}{edgelayer}
		\draw (3.center) to (2);
		\draw (2) to (1);
		\draw (1) to (0);
		\draw (0) to (5.center);
		\draw (1) to (4);
		\draw (6.center) to (7.center);
	\end{pgfonlayer}
\end{tikzpicture}&\hypertarget{r:EU}{(EU)} && \begin{tikzpicture}
	\begin{pgfonlayer}{nodelayer}
		\node [style=none] (0) at (0.7500001, 1) {};
		\node [style=none] (1) at (-0.7500001, -1) {};
		\node [style={H box}] (2) at (-0.7500001, 0.5000001) {};
		\node [style=rn] (3) at (-1.25, -0) {\footnotesize$~\alpha~$};
		\node [style={H box}] (4) at (-1.75, -0.5000001) {};
		\node [style=none] (5) at (-0.7500001, 1) {};
		\node [style={H box}] (6) at (-0.7500001, -0.5000001) {};
		\node [style=none] (7) at (-1.25, 0.7500001) {$\cdots$};
		\node [style=none] (8) at (0, -0) {$=$};
		\node [style=gn] (9) at (1.25, -0) {\footnotesize$~\alpha~$};
		\node [style=none] (10) at (1.75, -1) {};
		\node [style=none] (11) at (1.75, 1) {};
		\node [style=none] (12) at (-1.75, 1) {};
		\node [style={H box}] (13) at (-1.75, 0.5000001) {};
		\node [style=none] (14) at (-1.75, -1) {};
		\node [style=none] (15) at (0.7500001, -1) {};
		\node [style=none] (16) at (-1.25, -0.7500001) {$\cdots$};
		\node [style=none] (17) at (1.25, 0.7500001) {$\cdots$};
		\node [style=none] (18) at (1.25, -0.7500001) {$\cdots$};
	\end{pgfonlayer}
	\begin{pgfonlayer}{edgelayer}
		\draw [bend right, looseness=1.00] (3) to (4);
		\draw [bend left, looseness=1.00] (3) to (6);
		\draw (6) to (1.center);
		\draw (4) to (14.center);
		\draw [bend left, looseness=1.00] (3) to (13);
		\draw [bend right, looseness=1.00] (3) to (2);
		\draw (2) to (5.center);
		\draw (13) to (12.center);
		\draw [bend left=23, looseness=1.00] (9) to (0.center);
		\draw [bend right=23, looseness=1.00] (9) to (11.center);
		\draw [bend right=23, looseness=1.00] (9) to (15.center);
		\draw [bend left=23, looseness=1.00] (9) to (10.center);
	\end{pgfonlayer}
\end{tikzpicture}&\hypertarget{r:H}{(H)}\\
   &&&& \\
   \begin{tikzpicture}
	\begin{pgfonlayer}{nodelayer}
		\node [style=gn] (0) at (-1, -0) {$~\pi~$};
		\node [style=none] (1) at (-0.5, 0.75) {};
		\node [style=none] (2) at (-0.5, -0.75) {};
		\node [style=none] (3) at (0, -0) {$=$};
		\node [style=gn] (4) at (1, 0.25) {};
		\node [style=none] (5) at (1, 0.75) {};
		\node [style=rn] (6) at (1, -0.25) {};
		\node [style=none] (7) at (1, -0.75) {};
		\node [style=gn] (8) at (0.5, -0) {$~\pi~$};
	\end{pgfonlayer}
	\begin{pgfonlayer}{edgelayer}
		\draw (1.center) to (2.center);
		\draw (5.center) to (4);
		\draw (6) to (7.center);
	\end{pgfonlayer}
\end{tikzpicture}&\hypertarget{r:ZO}{(ZO)} && \begin{tikzpicture}
	\begin{pgfonlayer}{nodelayer}
		\node [style=gn] (0) at (-1.25, 0.7499999) {$~~\alpha~~$};
		\node [style=gn] (1) at (-0.2499998, 0.7499999) {$\alpha{+}\pi$};
		\node [style=gn] (2) at (1.25, 0.7499999) {$2\alpha{+}\pi$};
		\node [style=none] (3) at (1.25, -0.7500001) {};
		\node [style=rn] (4) at (-0.7500001, -0.2500001) {};
		\node [style=none] (5) at (-0.7499998, -0.7499998) {};
		\node [style=none] (6) at (0.2499998, -0) {$=$};
		\node [style=rn] (7) at (1.25, -0.2500001) {};
	\end{pgfonlayer}
	\begin{pgfonlayer}{edgelayer}
		\draw (0) to (4);
		\draw (1) to (4);
		\draw (4) to (5.center);
		\draw (7) to (3.center);
		\draw [bend right=45, looseness=1.00] (2) to (7);
		\draw [bend right=45, looseness=1.00] (7) to (2);
	\end{pgfonlayer}
\end{tikzpicture}&\hypertarget{r:SUP}{(SUP)}\\
   &&&& \\
   \hline
  \end{tabular}
 \caption[]{Rules for the ZX-calculus with scalars, augmented with the supplementarity rule \cite{supplementarity}. All of these rules also hold when flipped upside-down, or with the colours red and green swapped. The right-hand side of (IV) is an empty diagram. ($\cdots$) denote zero or more wires, while (\protect\rotatebox{45}{\raisebox{-0.4em}{$\cdots$}}) denote one or more wires.}
 \label{fig:ZX_rules}
\end{figure}
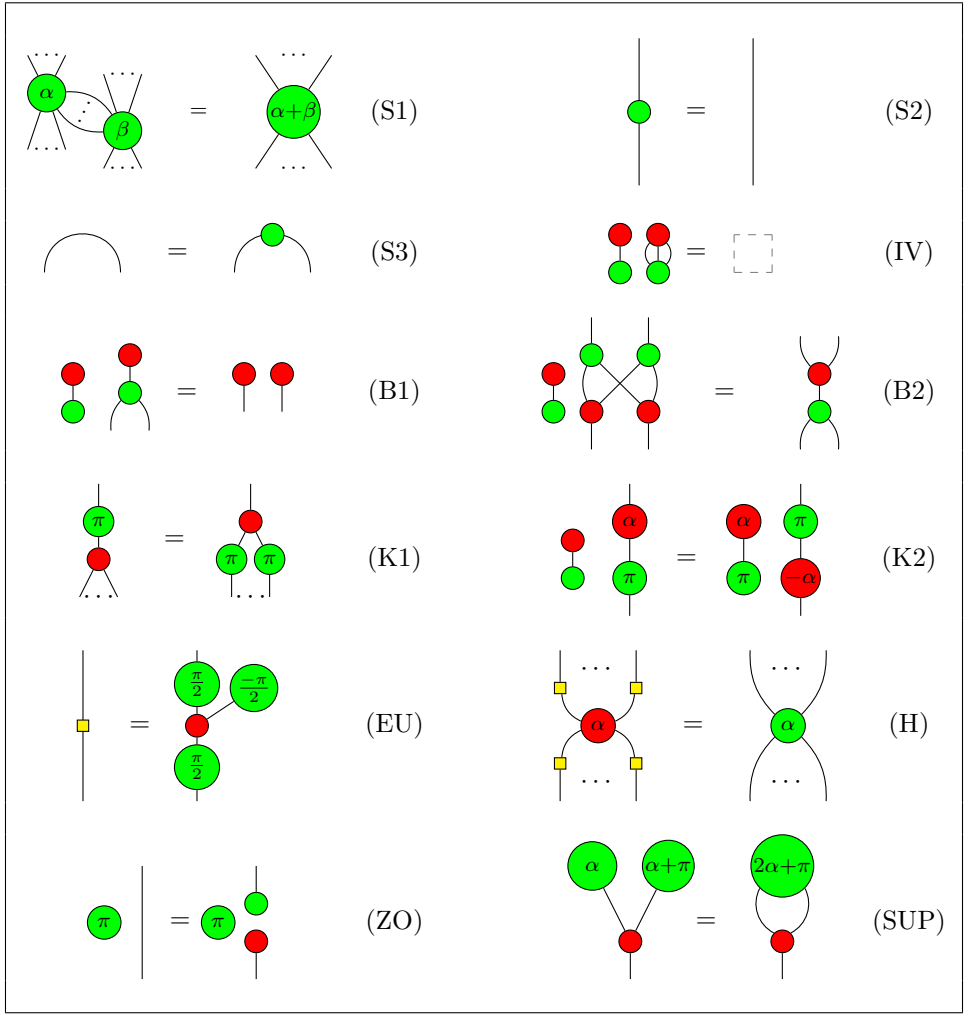

The initial set of axioms \cite{interacting} included the rules \sone, \stwo, \sthree, \bo, \bt, \ko, \kt and \h. The rule \eu has been proven to be necessary in \cite{euler-decomp} and the rules \zo and \iv result from \cite{scalar-completeness,simplified-stabilizer}. Finally, the rule \supp has been added in \cite{supplementarity}.

When we can show that a diagram $D_1$ is equal to another one, $D_2$, using a succession of equalities of this set, we write $ZX\vdash D_1 = D_2$. Given that the rules are sound, this implies that $\interp{D_1} = \interp{D_2}$. The rules can be applied to any subdiagram, meaning, for any diagram $D$:
\[ (ZX\vdash D_1=D_2)\implies \left\lbrace\begin{array}{ccc}
(ZX\vdash D_1\circ D = D_2\circ D) & \land & (ZX\vdash D\circ D_1 = D\circ D_2)\\
(ZX\vdash D_1\otimes D = D_2\otimes D) & \land & (ZX\vdash D\otimes D_1 = D\otimes D_2)
\end{array}\right. \]

\textbf{Scalars:} We will identify diagrams with 0 input and 0 output -- hence representing a $1\times 1$ matrix -- with scalars. We will not ignore them in this paper, while in some versions of the ZX-Calculus, the global phase or even all the scalars are ignored. Ignoring them would imply taking the risk of ignoring a zero scalar, which can lead to false statements -- if $ZX\vdash 0\times D_1 = 0 \times D_2$, we can not say that $ZX\vdash D_1 = D_2$. The first rules to palliate it appeared in \cite{scalar-completeness} and were simplified in \cite{simplified-stabilizer}.

\textbf{Only Topology Matters} is a paradigm -- to be taken as a rule -- stating that any wire of a ZX-diagram can be bent at will, without changing its semantics:
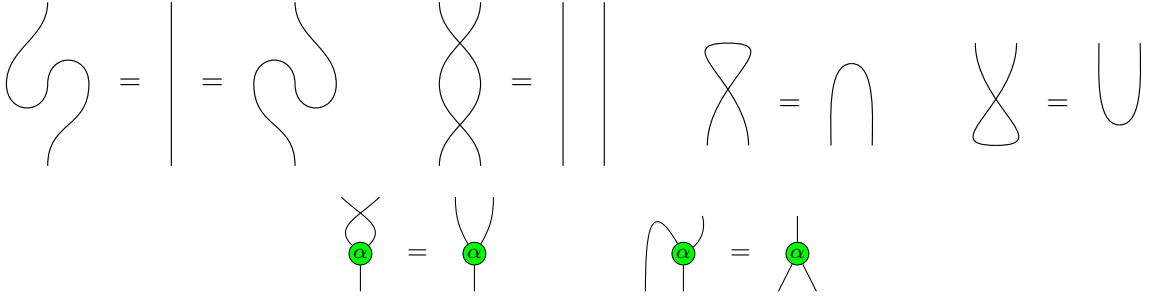
\[\fit{\begin{tikzpicture}
	\begin{pgfonlayer}{nodelayer}
		\node [style=none] (0) at (-5.5, -0) {=};
		\node [style=none] (1) at (-6.5, 1) {};
		\node [style=none] (2) at (-6.5, -1) {};
		\node [style=none] (3) at (-5, 1) {};
		\node [style=none] (4) at (-5, -1) {};
		\node [style=none] (5) at (-6.5, -0) {};
		\node [style=none] (6) at (-7, -0) {};
		\node [style=none] (7) at (-6, -0) {};
		\node [style=none] (8) at (-4, -0) {};
		\node [style=none] (9) at (-3.5, -1) {};
		\node [style=none] (10) at (-4.5, -0) {=};
		\node [style=none] (11) at (-3.5, 1) {};
		\node [style=none] (12) at (-3.5, -0) {};
		\node [style=none] (13) at (-3, -0) {};
		\node [style=none] (14) at (-1.75, 1) {};
		\node [style=none] (15) at (-1.75, -1) {};
		\node [style=none] (16) at (-1.25, 1) {};
		\node [style=none] (17) at (-1.25, -1) {};
		\node [style=none] (18) at (-1.75, -0) {};
		\node [style=none] (19) at (-1.25, -0) {};
		\node [style=none] (20) at (-0.2500001, 1) {};
		\node [style=none] (21) at (-0.2500001, -1) {};
		\node [style=none] (22) at (-0.7500003, -0) {=};
		\node [style=none] (23) at (0.2500001, 1) {};
		\node [style=none] (24) at (0.2500001, -1) {};
		\node [style=none] (25) at (1.5, -0.7499999) {};
		\node [style=none] (26) at (1.75, 0.5000002) {};
		\node [style=none] (27) at (2, -0.7499999) {};
		\node [style=none] (28) at (3.25, 0.2500001) {};
		\node [style=none] (29) at (3, -0.7499999) {};
		\node [style=none] (30) at (3.5, -0.7499999) {};
		\node [style=none] (31) at (2.5, -0.2500001) {=};
		\node [style=none] (32) at (6.25, 0.5000002) {};
		\node [style=none] (33) at (5, -0.7499999) {};
		\node [style=none] (34) at (6.75, 0.5000002) {};
		\node [style=none] (35) at (5.25, 0.5000002) {};
		\node [style=none] (36) at (5.75, -0.2500001) {=};
		\node [style=none] (37) at (6.5, -0.5000002) {};
		\node [style=none] (38) at (4.75, 0.5000002) {};
	\end{pgfonlayer}
	\begin{pgfonlayer}{edgelayer}
		\draw (3.center) to (4.center);
		\draw [in=90, out=-90, looseness=1.00] (1.center) to (6.center);
		\draw [bend right=90, looseness=2.00] (6.center) to (5.center);
		\draw [bend left=90, looseness=2.00] (5.center) to (7.center);
		\draw [in=90, out=-90, looseness=1.25] (7.center) to (2.center);
		\draw [in=90, out=-90, looseness=1.00] (11.center) to (13.center);
		\draw [bend left=90, looseness=2.00] (13.center) to (12.center);
		\draw [bend right=90, looseness=2.00] (12.center) to (8.center);
		\draw [in=90, out=-90, looseness=1.25] (8.center) to (9.center);
		\draw [style=none, in=90, out=-90, looseness=1.00] (14.center) to (19.center);
		\draw [style=none, in=90, out=-90, looseness=1.00] (19.center) to (15.center);
		\draw [style=none, in=90, out=-90, looseness=1.00] (16.center) to (18.center);
		\draw [style=none, in=90, out=-90, looseness=1.00] (18.center) to (17.center);
		\draw (20.center) to (21.center);
		\draw (23.center) to (24.center);
		\draw [in=0, out=90, looseness=1.50] (25.center) to (26.center);
		\draw [in=90, out=180, looseness=1.50] (26.center) to (27.center);
		\draw [in=0, out=90, looseness=0.75] (30.center) to (28.center);
		\draw [in=90, out=180, looseness=0.75] (28.center) to (29.center);
		\draw [in=0, out=-90, looseness=1.50] (38.center) to (33.center);
		\draw [in=-90, out=180, looseness=1.50] (33.center) to (35.center);
		\draw [in=0, out=-90, looseness=0.75] (34.center) to (37.center);
		\draw [in=-90, out=180, looseness=0.75] (37.center) to (32.center);
	\end{pgfonlayer}
\end{tikzpicture}}\]
\[\begin{tikzpicture}
	\begin{pgfonlayer}{nodelayer}
		\node [style=gn] (0) at (-2.75, -0) {$\alpha$};
		\node [style=none] (1) at (-3, 0.7499999) {};
		\node [style=none] (2) at (-2.5, 0.7499999) {};
		\node [style=none] (3) at (-2.75, -0.5) {};
		\node [style=none] (4) at (-2, -0) {=};
		\node [style=none] (5) at (-1.5, 0.7499999) {};
		\node [style=none] (6) at (-1.25, -0.5) {};
		\node [style=gn] (7) at (-1.25, -0) {$\alpha$};
		\node [style=none] (8) at (-0.9999999, 0.7499999) {};
		\node [style=none] (9) at (0.9999999, -0.5) {};
		\node [style=gn] (10) at (1.5, -0) {$\alpha$};
		\node [style=none] (11) at (2.25, -0) {=};
		\node [style=gn] (12) at (3, -0) {$\alpha$};
		\node [style=none] (13) at (1.75, 0.5) {};
		\node [style=none] (14) at (3.25, -0.5) {};
		\node [style=none] (15) at (1.5, -0.5) {};
		\node [style=none] (16) at (2.75, -0.5) {};
		\node [style=none] (17) at (3, 0.5) {};
	\end{pgfonlayer}
	\begin{pgfonlayer}{edgelayer}
		\draw [in=45, out=-45, looseness=1.25] (1.center) to (0);
		\draw [in=135, out=-135, looseness=1.25] (2.center) to (0);
		\draw (0) to (3.center);
		\draw [bend right=15, looseness=1.00] (5.center) to (7);
		\draw [bend left=15, looseness=1.00] (8.center) to (7);
		\draw (7) to (6.center);
		\draw [in=120, out=90, looseness=2.50] (9.center) to (10);
		\draw [bend left, looseness=1.00] (13.center) to (10);
		\draw (10) to (15.center);
		\draw (16.center) to (12);
		\draw (17.center) to (12);
		\draw (12) to (14.center);
	\end{pgfonlayer}
\end{tikzpicture}\]

\section{The $\boldsymbol{\frac{\pi}{4}}$-Fragment is not Complete}
\label{sec:pi_4}
In this section, we identify the following  simple equation \e, which is sound -- both sides of the equation represent the scalar $1$ -- but which cannot be derived from the rules of the ZX-Calculus expressed in Figure \ref{fig:ZX_rules}. 
\hypertarget{r:E}{}\[\begin{tikzpicture}
	\begin{pgfonlayer}{nodelayer}
		\node [style=rn] (0) at (-0.7500001, -0.32) {$\frac{\text{-}\pi}{4}$};
		\node [style=gn] (1) at (-0.7499998, 0.32) {$\,\frac{\pi}{4}\,$};
		\node [style=none] (2) at (0, -0) {=};
		\node [style=none] (3) at (0.500001, 0.25) {};
		\node [style=none] (4) at (0.500001, -0.25) {};
		\node [style=none] (5) at (1, 0.25) {};
		\node [style=none] (6) at (1, -0.25) {};
	\end{pgfonlayer}
	\begin{pgfonlayer}{edgelayer}
		\draw (0) to (1);
		\draw [style=dashed] (3.center) to (5.center);
		\draw [style=dashed] (5.center) to (6.center);
		\draw [style=dashed] (6.center) to (4.center);
		\draw [style=dashed] (4.center) to (3.center);
	\end{pgfonlayer}
\end{tikzpicture}\qquad (\text{E})\]

Since equation \e only involves angles multiple of $\frac\pi4$, it implies the incompleteness of the \frag{4} of the ZX-Calculus. 
In the following, 
we exhibit a simple invariant of the ZX-Calculus to prove that \e is not derivable, 
and then we show that \e subsumes two existing rules of the ZX-Calculus -- namely \iv and \zo~--, 
leading to a simpler -- \iv and \zo rules are replaced by \e~-- but more expressive ZX-Calculus that we call $ZX_E$. 

\subsection{A Graphical Invariant for the ZX-Calculus}

We introduce  a simple graphical quantity for ZX-diagrams, the parity of the number of odd-degree red dots plus the number of H-dots (yellow squares), formally defined as follows:

\begin{definition}Given a ZX-diagram $D:n\to m$, let $\invR{D}\in \{0,1\}$ be inductively defined as $\invR{D_1\otimes D_2} = \invR{D_1\circ D_2} = \invR{D_1}+\invR{D_2} \bmod 2$,   $\invR{\begin{tikzpicture}
	\begin{pgfonlayer}{nodelayer}
		\node [style=rn] (0) at (0, -0) {$~\alpha~$};
		\node [style=none] (1) at (-0.5000001, 0.7499999) {};
		\node [style=none] (2) at (0, 0.4999999) {$\cdots$};
		\node [style=none] (3) at (0.5000001, 0.7499999) {};
		\node [style=none] (4) at (-0.5000001, -0.7499999) {};
		\node [style=none] (5) at (0, -0.5) {$\cdots$};
		\node [style=none] (6) at (0.5000001, -0.7499999) {};
		\node [style=none] (7) at (0, 0.75) {$n$};
		\node [style=none] (8) at (0, -0.7499999) {$m$};
		\node [style=none] (9) at (0, 1) {};
		\node [style=none] (10) at (0, -1) {};
	\end{pgfonlayer}
	\begin{pgfonlayer}{edgelayer}
		\draw [style=none, bend right, looseness=1.00] (1.center) to (0);
		\draw [style=none, bend right, looseness=1.00] (0) to (3.center);
		\draw [style=none, bend right, looseness=1.00] (0) to (4.center);
		\draw [style=none, bend left, looseness=1.00] (0) to (6.center);
	\end{pgfonlayer}
\end{tikzpicture}}\hspace{-0.4em}=n+m\bmod 2$, $\invR{~\begin{tikzpicture}
	\begin{pgfonlayer}{nodelayer}
		\node [style={H box}] (0) at (0, 0) {};
		\node [style=none] (1) at (0, 0.5) {};
		\node [style=none] (2) at (0, -0.5) {};
	\end{pgfonlayer}
	\begin{pgfonlayer}{edgelayer}
		\draw (2.center) to (1.center);
	\end{pgfonlayer}
\end{tikzpicture}
~}\hspace{-0.4em}=1$, and $\invR{.}=0$ for all the other generators. 


\end{definition}

One can define similarly $\invG{.}$ as the parity of the number of odd-degree green dots plus the number of H-dots. Notice that for any scalar $D:0\to 0$, $\invG D+\invR D=0\bmod 2$, thanks to the well known degree sum formula which implies that the sum of the degree of the vertices of a graph is even. More generally, for any $D:n\to m$, $\invG D+\invR D=n+m \bmod 2$, which is clearly an invariant of the ZX-calculus since all the rules preserve the number of inputs/outputs. As a consequence, a rule preserves $\invG.$ if and only if it preserves $\invR.$. 

\begin{lemma}[Invariant]\label{lem:inv} For any ZX-diagram $D_1$,  if $\interp {D_1} \neq 0$ and $ZX\vdash D_1 = D_2$, then $\invR {D_1} = \invR{D_2}$. 
\end{lemma}
\begin{proof} Notice that all the rules in Figure \ref{fig:ZX_rules}, but \zo, preserve $\invR .$.
Since $\interp{D_1}\neq 0$, the scalar \begin{tikzpicture}
	\begin{pgfonlayer}{nodelayer}
		\node [style=gn] (0) at (0, -0) {$\pi$};
	\end{pgfonlayer}
\end{tikzpicture} cannot appear in any derivation transforming $D_1$ into $D_2$, thus $\zo$ is not applied, as a consequence $\invR {D_1} = \invR{D_2}$.
\end{proof}

\begin{proposition}\label{prop:incomplete} Equation \e is not derivable using the rules in Figure \ref{fig:ZX_rules}: $ZX\nvdash \e$.
\end{proposition}

\begin{proof}
The two diagrams of equation \e are non zero, and they differ for $\invR.$, so according to lemma \ref{lem:inv}, $ZX\nvdash \e$.
\end{proof}

Since the diagrams of equation \e are in the \frag4, it implies that the \frag4 of the ZX-calculus is not complete. 

\begin{remark}
\begin{tabularx}{\textwidth-\widthof{~~~Remark.~}}{cX@{}}
$\begin{tikzpicture}
	\begin{pgfonlayer}{nodelayer}
		\node [style=rn] (0) at (-0.7500001, -0.3) {$\frac{\text{-}\pi}{4}$};
		\node [style=gn] (1) at (-0.7499998, 0.3) {$\,\frac{\pi}{4}\,$};
		\node [style=none] (2) at (0.6, -0) {=};
		\node [style=none] (3) at (1, 0.2500001) {};
		\node [style=none] (4) at (1, -0.2500001) {};
		\node [style=none] (5) at (1.5, 0.2500001) {};
		\node [style=none] (6) at (1.5, -0.2500001) {};
		\node [style=gn] (7) at (0, 0.3) {$\,\frac{\pi}{4}\,$};
		\node [style=rn] (8) at (0, -0.3) {$\frac{\text{-}\pi}{4}$};
	\end{pgfonlayer}
	\begin{pgfonlayer}{edgelayer}
		\draw (0) to (1);
		\draw [style=dashed] (3.center) to (5.center);
		\draw [style=dashed] (5.center) to (6.center);
		\draw [style=dashed] (6.center) to (4.center);
		\draw [style=dashed] (4.center) to (3.center);
		\draw (8) to (7);
	\end{pgfonlayer}
\end{tikzpicture}$ & the ``doubled'' version of \e, contrary to it, is derivable in the ZX-Calculus.
\end{tabularx}
\end{remark}

By completeness of the \frag2, for any ZX-diagrams $D_1$ and $D_2$ in this particular fragment, if $\interp {D_1}=\interp {D_2}\neq 0$, then $\invR{D_1}=\invR{D_2}$. This property is obviously not true in the \frag4, equation \e being a counter example. However, this property is also satisfied by other, a priori not complete, fragments:

\begin{proposition}\label{prop:invfrag}For any $k\neq 0\bmod 4$ and any two diagrams $D_1,D_2$ with angles multiple of $\frac \pi k$, if $\interp {D_1} =\interp{D_2} \neq 0$ then $\invR{D_1} = \invR{D_2}$.  
\end{proposition}  

\begin{proof}
Given $k>0$ and  $D$ a ZX-diagram with angles multiple of $\frac \pi k$, one can show, by induction on $D$, that $(\sqrt2^{\invR{D}})\interp D$ is a matrix whose entries are in $\QQ[e^{\frac{i\pi}k}]$, the smallest subfield of $\CCC$ which contains $e^{\frac{i\pi}k}$. Since there is a non-zero entry in $\interp {D_1}$, there exist $q_1,q_2\in  \QQ[e^{\frac{i\pi}k}]$ such that $\sqrt2^{\invR{D_1}}q_1 =\sqrt2^{\invR{D_2}}q_2 \neq 0$, so $\sqrt 2^{(\invR{D_1}-\invR{D_2})} \in  \QQ[e^{\frac{i\pi}k}]$. Suppose $\sqrt{2} \in \QQ[e^{\frac{i\pi}k}]$:
\begin{itemize}
\item[$i$.] If $k=2 \bmod 4$,  then $i=e^{\frac{i\pi}2}\in \QQ[e^{\frac{i\pi}k}]$, so
$e^{\frac{-i\pi}{4}} =  \sqrt{2} \frac{1-i}2 \in \QQ[e^{\frac{i\pi}{k}}]$.
Therefore
$e^{\frac{i\pi}{2k}} = (e^{\frac{i\pi}k})^{\frac{k+2}{4}} \times e^{\frac{-i\pi}4} \in
\QQ[e^{\frac{i\pi}k}]$. 
This implies $\QQ[e^{\frac{i\pi}{2k}}] = \QQ[e^{\frac{i\pi}k}]$ which is not possible as
they are vector spaces over $\QQ$ of dimension respectively $\varphi(4k)=2\varphi(2k)$ and
$\varphi(2k)$ where $\varphi$ is Euler's totient function, and $\varphi(2k)\neq 0$.
\item[$ii$.] If $k$ is odd then $2k=2\bmod 4$. Moreover,  since   $\QQ[e^{i\frac{\pi}{k}}]\subseteq \QQ[e^{i\frac{\pi}{2k}}]$,  $\sqrt{2} \in \QQ[e^{i\frac{\pi}{2k}}]$ which is impossible according to the previous case ($i$).
\end{itemize}
Thus $\sqrt 2\notin \QQ[e^{\frac{i\pi}k}]$ when $k\neq 0\bmod 4$, so $\invR{D_1} = \invR{D_2}$.
\end{proof}

\subsection{A Simpler and More Expressive ZX-calculus}

Equation \e cannot be derived in the ZX-calculus (proposition \ref{prop:incomplete}), as a consequence we propose to add this equation \e as a rule of the language to make it more expressive. We show in the following that the introduction of this new rule makes the two scalar rules \zo and \iv obsolete, leading to a language with less rules than the one define in Figure  \ref{fig:ZX_rules}.

Let us define $ZX_E = \{\e\}\cup ZX\setminus\{\iv,\zo\}$. First, notice that, thanks to \cite{simplified-stabilizer}, the so-called Hopf law is derivable from $ZX\setminus\{\iv,\zo\}$, and hence from $ZX_E$:
\begin{lemma}
\hypertarget{r:HL}{}\[ZX\setminus\{\iv,\zo\} \vdash~~ \begin{tikzpicture}
	\begin{pgfonlayer}{nodelayer}
		\node [style=rn] (0) at (0, 0.25) {};
		\node [style=rn] (1) at (-1, 0.2) {};
		\node [style=gn] (2) at (-1, -0.2) {};
		\node [style=rn] (3) at (-0.5, 0.2) {};
		\node [style=gn] (4) at (-0.5, -0.2) {};
		\node [style=gn] (5) at (0, -0.25) {};
		\node [style=none] (6) at (0, 0.6) {};
		\node [style=none] (7) at (0, -0.6) {};
		\node [style=none] (8) at (1.5, 0.6) {};
		\node [style=rn] (9) at (1.5, 0.25) {};
		\node [style=none] (10) at (1.5, -0.6) {};
		\node [style=gn] (11) at (1.5, -0.25) {};
		\node [style=none] (12) at (0.75, -0) {$=$};
	\end{pgfonlayer}
	\begin{pgfonlayer}{edgelayer}
		\draw (1) to (2);
		\draw (3) to (4);
		\draw [bend right=45, looseness=1.00] (0) to (5);
		\draw [bend left=45, looseness=1.00] (0) to (5);
		\draw (0) to (6.center);
		\draw (5) to (7.center);
		\draw (9) to (8.center);
		\draw (11) to (10.center);
	\end{pgfonlayer}
\end{tikzpicture} \qquad (\textnormal{HL})\]
\end{lemma}

\begin{proposition}
\iv is derivable from $ZX_E$.
\end{proposition}
\begin{proof}
Using \e, \bo, \hlaw, \stwo and \sone:\\\vspace{-0.5em}\\
$\begin{tikzpicture}
	\begin{pgfonlayer}{nodelayer}
		\node [style=rn] (0) at (-5.75, -0) {};
		\node [style=rn] (1) at (-6.25, -0) {};
		\node [style=gn] (2) at (-6.25, -0.5000001) {};
		\node [style=gn] (3) at (-5.75, -0.5000001) {};
		\node [style=none] (4) at (-5, -0.25) {=};
		\node [style=gn] (5) at (-4.25, -0.5000001) {};
		\node [style=rn] (6) at (-3.75, -0) {};
		\node [style=rn] (7) at (-4.25, -0) {};
		\node [style=gn] (8) at (-3.75, -0.5000001) {};
		\node [style=gn] (9) at (-3, -0.6) {\,$\frac\pi 4$\,};
		\node [style=rn] (10) at (-3, 0.1) {$\frac {\text{-}\pi}4$};
		\node [style=none] (11) at (-2.25, -0.25) {=};
		\node [style=gn] (12) at (0, -0.7) {\,$\frac\pi 4$\,};
		\node [style=rn] (13) at (0, -0) {$\frac {\text{-}\pi}4$};
		\node [style=gn] (14) at (-1, -0.5) {};
		\node [style=rn] (15) at (-1, -0) {};
		\node [style=rn] (16) at (-1.5, -0.5) {};
		\node [style=gn] (17) at (-1.5, -1) {};
		\node [style=rn] (18) at (-0.5, 1) {};
		\node [style=gn] (19) at (-0.5, 0.5) {};
		\node [style=rn] (20) at (-1.5, 0.5) {};
		\node [style=gn] (21) at (-1.5, -0) {};
		\node [style=gn] (22) at (1.5, -0.5) {};
		\node [style=gn] (23) at (2, 0.5) {};
		\node [style=rn] (24) at (2.5, -0) {$\frac {\text{-}\pi}4$};
		\node [style=none] (25) at (0.75, -0.25) {=};
		\node [style=rn] (26) at (2, 1) {};
		\node [style=gn] (27) at (2.5, -0.7) {\,$\frac\pi 4$\,};
		\node [style=rn] (28) at (1.5, -0) {};
		\node [style=none] (29) at (3.25, -0.25) {=};
		\node [style=rn] (30) at (4, 0.1) {$\frac {\text{-}\pi}4$};
		\node [style=gn] (31) at (4, -0.6) {\,$\frac\pi 4$\,};
		\node [style=none] (32) at (5.749999, -0) {};
		\node [style=none] (33) at (5.249999, -0) {};
		\node [style=none] (34) at (4.75, -0.25) {$=$};
		\node [style=none] (35) at (5.749999, -0.5) {};
		\node [style=none] (36) at (5.249999, -0.5) {};
	\end{pgfonlayer}
	\begin{pgfonlayer}{edgelayer}
		\draw [bend left=45, looseness=1.00] (0) to (3);
		\draw (1) to (2);
		\draw [bend right=45, looseness=1.00] (0) to (3);
		\draw (0) to (3);
		\draw [bend left=45, looseness=1.00] (6) to (8);
		\draw (7) to (5);
		\draw [bend right=45, looseness=1.00] (6) to (8);
		\draw (6) to (8);
		\draw (10) to (9);
		\draw [bend left=45, looseness=1.00] (15) to (14);
		\draw (16) to (17);
		\draw [bend right=45, looseness=1.00] (15) to (14);
		\draw (15) to (14);
		\draw (13) to (12);
		\draw (18) to (19);
		\draw (19) to (15);
		\draw (19) to (13);
		\draw (20) to (21);
		\draw (28) to (22);
		\draw (24) to (27);
		\draw (26) to (23);
		\draw (23) to (28);
		\draw (23) to (24);
		\draw (30) to (31);
		\draw [color=gray, dashed] (33.center) to (36.center);
		\draw [color=gray, dashed] (36.center) to (35.center);
		\draw [color=gray, dashed] (35.center) to (32.center);
		\draw [color=gray, dashed] (32.center) to (33.center);
	\end{pgfonlayer}
\end{tikzpicture}$
\end{proof}

\begin{proposition}
\label{prop:zo-deducible}
\zo is derivable from $ZX_E$.
\end{proposition}
\begin{proof}
In appendix at page \pageref{prf:zo-deducible}.
\end{proof}
Hence $ZX_E \vdash \iv,\zo$.

\begin{remark}
The other rules of the language remain, a priori, necessary in the presence of equation \e. In particular the supplementarity which has been recently proved to be necessary  in ZX \cite{supplementarity} is necessary in $ZX_E$: one can prove using the interpretation $\interp{.}^{\sharp}_{k,l}$ -- defined in \cite{supplementarity} --  with $k=3$ and $l=8$ that $ZX_E\setminus\{\supp\}\nvdash \supp$. 
\end{remark}

\begin{remark}One may want to generalise equation \e, replacing the particular angles $\pm\frac\pi4$ by some generic angle $\alpha$. Proposition \ref{prop:invfrag} is a strong evidence that such a generalisation is not possible and that the language requires at least one rule which is specific to the  $\frac \pi4$ angle. 
\end{remark}

\section{Cyclotomic Supplementarity}
\label{sec:generalised-supplementarity}

\subsection{Generalisation of (SUP)}
The concept of supplementarity in quantum diagram reasoning has been first introduced by Coecke and Edwards \cite{w-in-zx}, turned into a simple but necessary rule (SUP in Figure \ref{fig:ZX_rules}) in \cite{supplementarity}. Roughly speaking, supplementarity consists in merging two dots sharing the same neighbour when the difference of their angles is $\pi$, i.e.~when the two angles are antipodal. We generalize this concept to cyclotomic supplementarity as follows: for any $n\in \mathbb N^*$, $n$ dots sharing the same neighbour can be merged when their angles divide the circle into equal parts (cyclotomy), i.e.~when their angles are of the form $\alpha+\frac{2k\pi}{n}\text{ for }k\in\interp{0;n-1}$:
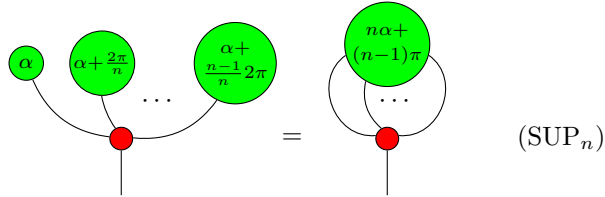
\hypertarget{r:SUP_n}{}\[\begin{tikzpicture}
	\begin{pgfonlayer}{nodelayer}
		\node [style=rn] (0) at (-1.75, -0) {};
		\node [style=none] (1) at (0.5, -0) {=};
		\node [style=none] (2) at (-1.75, -0.75) {};
		\node [style=gn] (3) at (-2, 1) {\scriptsize $\alpha{+}\frac{2\pi}{n}$};
		\node [style=gn, align=center] (4) at (-0.25, 1) {\scriptsize $\alpha{+}$\\ \scriptsize $\frac{n-1}{n}2\pi$};
		\node [style=gn] (5) at (-3, 1) {$~\alpha~$};
		\node [style=none] (6) at (-1.25, 0.5) {$\cdots$};
		\node [style=none] (7) at (1.75, -0.75) {};
		\node [style=rn] (8) at (1.75, -0) {};
		\node [style=gn, align=center] (9) at (1.75, 1.25) {\scriptsize $n\alpha$+\\\scriptsize $(n{-}1)\pi$};
		\node [style=none] (10) at (1.875, 0.5) {$\cdots$};
	\end{pgfonlayer}
	\begin{pgfonlayer}{edgelayer}
		\draw [bend right, looseness=1.00] (5) to (0);
		\draw [bend left=15, looseness=0.75] (0) to (3);
		\draw [bend left, looseness=1.00] (4) to (0);
		\draw (0) to (2.center);
		\draw [bend right=75, looseness=1.25] (9) to (8);
		\draw (8) to (7.center);
		\draw [bend right, looseness=1.00] (9) to (8);
		\draw [bend right=75, looseness=1.25] (8) to (9);
	\end{pgfonlayer}
\end{tikzpicture}\qquad(\text{SUP}_n)\]
Notice that there are $n$ green dots in the left diagram, and $n$ parallel wires in the right diagram.

Any of these equations is valid for the standard interpretation of ZX-diagrams:  
\begin{proposition}
\label{prop:suppn-sound}
\suppn is sound.
\end{proposition}

\begin{proof}
In appendix at page \pageref{prf:suppn-sound}.
\end{proof}

Cyclotomic supplementarity has a generalisation: the green dots can be merged not only when they share the same neighbour, but also when they share the same neighbourhood. It leads to the notion of cyclotomic twins, which generalise the notion of antiphase twins [13]:

\begin{definition}[Cyclotomic Twins]
$n$ dots in a ZX-diagram are cyclotomic twins if:
\begin{itemize}
\item they have the same colour
\item their angles divide the circle into equal parts $\left(\alpha+\frac{2k\pi}{n}\text{ for }k\in\interp{0;n-1}\right)$
\item they have the same neighbourhood: for any vertex, the number of wires connecting it to any of the twins is the same
\end{itemize}
\end{definition}

\begin{proposition}[Cyclotomic Twins and Supplementarity]
\label{prop:cyclotomic-neighbourhood}
With $ZX\cup\{\suppn\}_{n\in\mathbb{N}}$, cyclotomic twins can be merged.
\end{proposition}

\begin{proof}
In appendix at page \pageref{prf:cyclotomic-neighbourhood}.
\end{proof}

The rest of the section is dedicated to the structures of this family of equations: we show that \suppc n is necessary when $n$ is an odd prime number and that \suppc n can be derived when $n$ is not prime. As a consequence, we exhibit a countable family of equations that cannot be derived in the ZX-calculus.

\subsection{The Set of Supplementarity Rules for Prime Numbers}
It is not necessary to define the supplementarity rules for all numbers $n\in\mathbb{N}$ as axioms. For instance, we will prove that their restriction to the set of prime numbers is enough to show all the others.

Let $\mathbb{P}$ be the set of prime numbers.

\begin{theorem}
\label{th:supp-n-for-primes}~~
\begin{itemize}
\item $\forall p,q\in\mathbb{N}^*,~~ ZX_E\cup\{\suppc{p},\suppc{q}\}\vdash \suppc{pq}$
\item $\forall n\in \mathbb{N}^*,~~ ZX_E\cup\{\suppc{p}\}_{p\in \mathbb{P}}\vdash \suppn$
\item $\forall p \in \mathbb{P},p\geq 3,~~ ZX_E\cup\{\suppc{q}\}_{q\in\mathbb{P}\setminus\{p\}}\nvdash \suppc{p}$
\end{itemize}
\end{theorem}

\begin{proof}
\phantomsection\label{prf:supp-n-for-primes}~~\\
\textbf{First statement:} If $n$ is not prime, its supplementarity can be derived. Indeed, suppose $n$ can be decomposed in two numbers $p$ and $q$ ($n=pq$), for which we know the supplementarity rule.
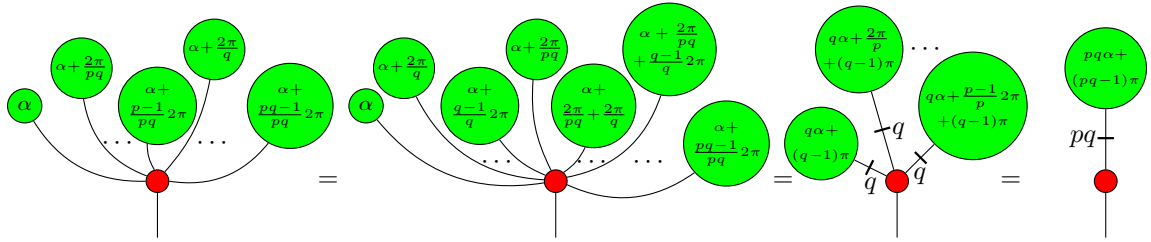
\[\fit{\begin{tikzpicture}
	\begin{pgfonlayer}{nodelayer}
		\node [style=rn] (0) at (-6.25, -0.25) {};
		\node [style=none] (1) at (-4, -0.25) {=};
		\node [style=none] (2) at (-6.25, -1) {};
		\node [style=gn] (3) at (-7.25, 1.25) {\tiny $\alpha{+}\frac{2\pi}{pq}$};
		\node [style=gn, align=center] (4) at (-6.25, 0.7500001) {\tiny $\alpha{+}$\\ \tiny $\frac{p-1}{pq}2\pi$};
		\node [style=gn] (5) at (-8, 0.7500001) {$~\alpha~$};
		\node [style=none] (6) at (-6.75, 0.25) {$\cdots$};
		\node [style=none] (7) at (-5.5, 0.25) {$\cdots$};
		\node [style=gn] (8) at (-5.5, 1.5) {\tiny $\alpha{+}\frac{2\pi}{q}$};
		\node [style=gn, align=center] (9) at (-4.5, 0.7500001) {\tiny $\alpha{+}$\\ \tiny $\frac{pq-1}{pq}2\pi$};
		\node [style=gn] (10) at (-3, 1.25) {\tiny $\alpha{+}\frac{2\pi}{q}$};
		\node [style=none] (11) at (-0.5000001, -0) {$\cdots$};
		\node [style=gn] (12) at (-3.5, 0.7500001) {$~\alpha~$};
		\node [style=gn, align=center] (13) at (-2, 0.7500001) {\tiny $\alpha{+}$\\ \tiny $\frac{q-1}{q}2\pi$};
		\node [style=none] (14) at (-1.75, -0) {$\cdots$};
		\node [style=gn, align=center] (15) at (1.25, 0.2500001) {\tiny $\alpha{+}$\\ \tiny $\frac{pq-1}{pq}2\pi$};
		\node [style=rn] (16) at (-1, -0.25) {};
		\node [style=none] (17) at (-1, -1) {};
		\node [style=gn] (18) at (-1.25, 1.5) {\tiny $\alpha{+}\frac{2\pi}{pq}$};
		\node [style=gn, align=center] (19) at (-0.5000002, 0.7499999) {\tiny $\alpha{+}$\\ \tiny $\frac{2\pi}{pq}{+}\frac{2\pi}{q}$};
		\node [style=gn, align=center] (20) at (0.5000002, 1.5) {\tiny $\alpha+\frac{2\pi}{pq}$\\ \tiny $+\frac{q-1}{q}2\pi$};
		\node [style=none] (21) at (0.2499997, -0) {$\cdots$};
		\node [style=none] (22) at (3.5, -1) {};
		\node [style=gn, align=center] (23) at (2.5, 0.2500001) {\tiny $q\alpha{+}$ \\ \tiny $(q{-}1)\pi$};
		\node [style=gn, align=center] (24) at (3, 1.5) {\tiny $q\alpha{+}\frac{2\pi}{p}$ \\ \tiny ${+}(q{-}1)\pi$};
		\node [style=rn] (25) at (3.5, -0.25) {};
		\node [style=none] (26) at (3.9, 1.5) {$\cdots$};
		\node [style=gn, align=center] (27) at (4.5, 0.7500001) {\tiny $q\alpha{+}\frac{p{-}1}{p}2\pi$ \\ \tiny ${+}(q{-}1)\pi$};
		\node [style=none] (28) at (2, -0.25) {=};
		\node [style=none] (29) at (6.25, -1) {};
		\node [style=rn] (30) at (6.25, -0.25) {};
		\node [style=none] (31) at (5, -0.25) {=};
		\node [style=gn, align=center] (32) at (6.25, 1.25) {\tiny $pq\alpha{+}$ \\ \tiny $(pq{-}1)\pi$};
	\end{pgfonlayer}
	\begin{pgfonlayer}{edgelayer}
		\draw [bend right, looseness=1.00] (5) to (0);
		\draw [bend left, looseness=1.00] (0) to (3);
		\draw [bend right=15, looseness=1.00] (4) to (0);
		\draw (0) to (2.center);
		\draw [bend left=15, looseness=1.00] (8) to (0);
		\draw [bend right, looseness=1.00] (0) to (9);
		\draw [bend right, looseness=1.00] (12) to (16);
		\draw [bend left=15, looseness=1.00] (16) to (13);
		\draw [bend right=15, looseness=1.00] (18) to (16);
		\draw (16) to (17.center);
		\draw [bend right, looseness=1.00] (10) to (16);
		\draw [bend right, looseness=1.00] (16) to (15);
		\draw [bend left=15, looseness=1.00] (19) to (16);
		\draw [bend right, looseness=1.00] (16) to (20);
		\draw (25) to (22.center);
		\draw (30) to (29.center);
		\path (23) edge[style=tickedge] node[below] {$q$} (25);
		\path (24) edge[style=tickedge] node[right] {$q$} (25);
		\path (27) edge[style=tickedge] node[below] {$q$} (25);
		\path (32) edge[style=tickedge] node[left] {$pq$} (30);
	\end{pgfonlayer}
\end{tikzpicture}}\]
with $p$-ticked edge representing $p$ parallel wires. The first equality is just a rearranging of the branches, the second uses \suppc{q} $p$ times and the last one exploits Proposition \ref{prop:cyclotomic-neighbourhood} with $p(q\alpha+(q-1)\pi)+(p-1)\pi = pq\alpha+(pq-1)\pi$.

\noindent\textbf{Second Statement:}
As a direct consequence of the previous statement, since \suppc{1} is trivial, the supplementarity rules for prime numbers are enough to derive all the others.

\noindent\textbf{Third Statement:} Let $p\in \mathbb{P}$ and $p\geq 3$. Let us consider the interpretation $\interp{.}_{p^2}$ which amounts to multiplying all the angles of a diagram by $p^2$.

\begin{itemize}
\item The interpretation $\interp{.}_{p^2}$ coincides with the interpretation $\interp{.}^{\sharp}_{1,p^2-1}$ defined in \cite{supplementarity}. As stated in this article, since the first parameter is odd and the second one is even, all the rules of ZX$\setminus\{\suppt\}$ hold.
\item The rule \e also holds. Indeed, $p$ is odd, and whether $p \mod 8$ is $1,3,5$ or $7$, $p^2 \mod 8 = 1$, so $p^2\frac{\pi}{4}=\frac{\pi}{4}\mod 2\pi$.
\item The rule \suppc{q} when $q\in \mathbb{P}$, $q\neq p$ holds, since $gcd(p^2, q) =1$:
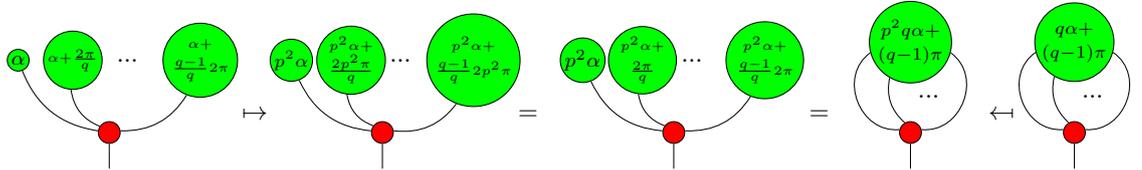
\[\fit{\begin{tikzpicture}
	\begin{pgfonlayer}{nodelayer}
		\node [style=rn] (0) at (-6.5, -0.5) {};
		\node [style=none] (1) at (-0.75, -0.25) {=};
		\node [style=none] (2) at (-6.5, -1) {};
		\node [style=gn] (3) at (-7, 0.5) {\tiny $\alpha{+}\frac{2\pi}{q}$};
		\node [style=gn, align=center] (4) at (-5.25, 0.5) {\tiny $\alpha{+}$\\ \tiny$\frac{q-1}{q}2\pi$};
		\node [style=gn] (5) at (-7.75, 0.5) {$\alpha$};
		\node [style=none] (6) at (-6.25, 0.5) {...};
		\node [style=none] (7) at (4.5, -1) {};
		\node [style=rn] (8) at (4.5, -0.5) {};
		\node [style=gn, align=center] (9) at (4.5, 0.75) {\scriptsize $p^2q\alpha+$ \\ \scriptsize $(q{-}1)\pi$};
		\node [style=none] (10) at (4.75, -0) {...};
		\node [style=none] (11) at (-4.5, -0.25) {$\mapsto$};
		\node [style=none] (12) at (5.75, -0.25) {$\mapsfrom$};
		\node [style=none] (13) at (-2.75, -1) {};
		\node [style=gn, align=center, xshift={2 pt}] (14) at (-3.25, 0.5) {\tiny $p^2\alpha{+}$ \\ \tiny $\frac{2p^2\pi}{q}$};
		\node [style=gn] (15) at (-4, 0.5) {\scriptsize $p^2\alpha$};
		\node [style=none] (16) at (-2.5, 0.5) {...};
		\node [style=rn] (17) at (-2.75, -0.5) {};
		\node [style=gn, align=center] (18) at (-1.5, 0.5) {\tiny $p^2\alpha+$\\\tiny $\frac{q-1}{q}2p^2\pi$};
		\node [style=none] (19) at (1.25, -1) {};
		\node [style=gn, align=center, xshift={2 pt}] (20) at (0.75, 0.5) {\tiny $p^2\alpha{+}$ \\ \tiny $\frac{2\pi}{q}$};
		\node [style=gn] (21) at (0, 0.5) {$p^2\alpha$};
		\node [style=none] (22) at (1.5, 0.5) {...};
		\node [style=rn] (23) at (1.25, -0.5) {};
		\node [style=gn, align=center] (24) at (2.5, 0.5) {\tiny $p^2\alpha+$\\\tiny $\frac{q-1}{q}2\pi$};
		\node [style=none] (25) at (3.25, -0.25) {=};
		\node [style=rn] (26) at (6.75, -0.5) {};
		\node [style=none] (27) at (7, -0) {...};
		\node [style=none] (28) at (6.75, -1) {};
		\node [style=gn, align=center] (29) at (6.75, 0.75) {\scriptsize $q\alpha+$ \\ \scriptsize $(q{-}1)\pi$};
	\end{pgfonlayer}
	\begin{pgfonlayer}{edgelayer}
		\draw [bend right, looseness=1.00] (5) to (0);
		\draw [bend left, looseness=1.00] (0) to (3);
		\draw [bend left, looseness=1.00] (4) to (0);
		\draw (0) to (2.center);
		\draw [bend right=75, looseness=1.25] (9) to (8);
		\draw (8) to (7.center);
		\draw [bend right, looseness=1.00] (9) to (8);
		\draw [bend right=75, looseness=1.25] (8) to (9);
		\draw [bend right, looseness=1.00] (15) to (17);
		\draw [bend left, looseness=1.00] (17) to (14);
		\draw [bend left, looseness=1.00] (18) to (17);
		\draw (17) to (13.center);
		\draw [bend right, looseness=1.00] (21) to (23);
		\draw [bend left, looseness=1.00] (23) to (20);
		\draw [bend left, looseness=1.00] (24) to (23);
		\draw (23) to (19.center);
		\draw [bend right=75, looseness=1.25] (29) to (26);
		\draw (26) to (28.center);
		\draw [bend right, looseness=1.00] (29) to (26);
		\draw [bend right=75, looseness=1.25] (26) to (29);
	\end{pgfonlayer}
\end{tikzpicture}}\]
\item The rule \suppc{p} does not hold:
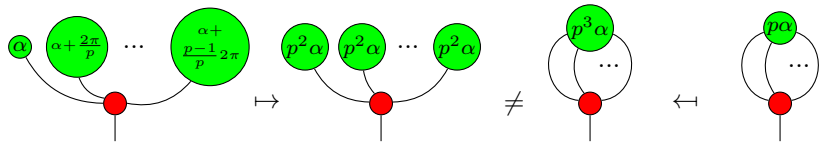
\[\begin{tikzpicture}
	\begin{pgfonlayer}{nodelayer}
		\node [style=rn] (0) at (-5, -0.25) {};
		\node [style=none] (1) at (0.25, -0.25) {$\neq$};
		\node [style=none] (2) at (-5, -0.75) {};
		\node [style=gn] (3) at (-5.5, 0.5) {\tiny $\alpha{+}\frac{2\pi}{p}$};
		\node [style=gn, align=center] (4) at (-3.75, 0.5) {\tiny $\alpha{+}$\\ \tiny $\frac{p-1}{p}2\pi$};
		\node [style=gn] (5) at (-6.25, 0.5) {$\alpha$};
		\node [style=none] (6) at (-4.75, 0.5) {...};
		\node [style=none] (7) at (1.25, -0.75) {};
		\node [style=rn] (8) at (1.25, -0.25) {};
		\node [style=gn] (9) at (1.25, 0.75) { $p^3\alpha$};
		\node [style=none] (10) at (1.5, 0.25) {...};
		\node [style=none] (11) at (-3, -0.25) {$\mapsto$};
		\node [style=none] (12) at (2.5, -0.25) {$\mapsfrom$};
		\node [style=rn] (13) at (3.75, -0.25) {};
		\node [style=none] (14) at (4, 0.25) {...};
		\node [style=none] (15) at (3.75, -0.75) {};
		\node [style=gn] (16) at (3.75, 0.75) { $p\alpha$};
		\node [style=gn] (17) at (-0.5, 0.5) {$p^2\alpha$};
		\node [style=none, xshift={3 pt}] (18) at (-1.25, 0.5) {...};
		\node [style=none] (19) at (-1.5, -0.75) {};
		\node [style=gn] (20) at (-2.5, 0.5) {$p^2\alpha$};
		\node [style=rn] (21) at (-1.5, -0.25) {};
		\node [style=gn] (22) at (-1.75, 0.5) {$p^2\alpha$};
	\end{pgfonlayer}
	\begin{pgfonlayer}{edgelayer}
		\draw [bend right, looseness=1.00] (5) to (0);
		\draw [bend left, looseness=1.00] (0) to (3);
		\draw [bend left, looseness=1.00] (4) to (0);
		\draw (0) to (2.center);
		\draw [bend right=75, looseness=1.25] (9) to (8);
		\draw (8) to (7.center);
		\draw [bend right, looseness=1.00] (9) to (8);
		\draw [bend right=75, looseness=1.25] (8) to (9);
		\draw [bend right=75, looseness=1.25] (16) to (13);
		\draw (13) to (15.center);
		\draw [bend right, looseness=1.00] (16) to (13);
		\draw [bend right=75, looseness=1.25] (13) to (16);
		\draw [bend right, looseness=1.00] (20) to (21);
		\draw [bend left=15, looseness=1.00] (21) to (22);
		\draw [bend left, looseness=1.00] (17) to (21);
		\draw (21) to (19.center);
	\end{pgfonlayer}
\end{tikzpicture}\]
Indeed, when $\alpha=0$, for instance, using on the left side $p-1$ times \bo and \iv, and on the right side \iv and $\frac{p-1}{2}$ times the Hopf law \hlaw, since $p\geq 3$:
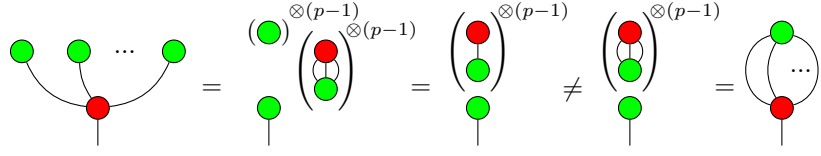
\[\begin{tikzpicture}
	\begin{pgfonlayer}{nodelayer}
		\node [style=gn] (0) at (-5, 0.5) {};
		\node [xshift=3 pt, style=none] (1) at (-5.75, 0.5) {...};
		\node [style=none] (2) at (-6, -0.75) {};
		\node [style=gn] (3) at (-7, 0.5) {};
		\node [style=rn] (4) at (-6, -0.25) {};
		\node [style=gn] (5) at (-6.25, 0.5) {};
		\node [style=none] (6) at (-4.5, -0) {=};
		\node [style=gn] (7) at (-3.75, -0.25) {};
		\node [style=none] (8) at (-3.75, -0.75) {};
		\node [style=gn] (9) at (-3.75, 0.75) {};
		\node [style=gn] (10) at (-1, -0.25) {};
		\node [style=gn] (11) at (-1, 0.25) {};
		\node [style=rn] (12) at (-1, 0.75) {};
		\node [style=none] (13) at (-1, -0.75) {};
		\node [style=none] (14) at (-1.75, -0) {=};
		\node [style=none] (15) at (0.25, -0) {$\neq$};
		\node [style=rn] (16) at (3, -0.25) {};
		\node [style=none] (17) at (3, -0.75) {};
		\node [style=gn] (18) at (3, 0.75) {};
		\node [style=none] (19) at (3.25, 0.25) {...};
		\node [style=gn] (20) at (1, -0.25) {};
		\node [style=none] (21) at (2.25, -0) {=};
		\node [style=none] (22) at (1, -0.75) {};
		\node [style=none, anchor=west, xshift=-4pt] (23) at (-2.75, 0.25) {$\left.\vphantom{\rule{1pt}{1.6em}}\right)$};
		\node [style=none, anchor=east, xshift=4pt] (24) at (-3.25, 0.25) {$\left(\vphantom{\rule{1pt}{1.6em}}\right.$};
		\node [style=none, anchor=west] (25) at (-2.75, 0.75) {$\vphantom{.}^{\otimes (p{-}1)}$};
		\node [xshift=4pt, anchor=east, style=none] (26) at (-4, 0.75) {$\left(\vphantom{\rule{1pt}{0.6em}}\right.$};
		\node [xshift=-4pt, anchor=west, style=none] (27) at (-3.5, 0.75) {$\left.\vphantom{\rule{1pt}{0.6em}}\right)$};
		\node [style=none, anchor=west] (28) at (-3.5, 1) {$\vphantom{.}^{\otimes (p{-}1)}$};
		\node [style=rn] (29) at (-3, 0.5) {};
		\node [style=gn] (30) at (-3, -0) {};
		\node [xshift=4pt, anchor=east, style=none] (31) at (-1.25, 0.5) {$\left(\vphantom{\rule{1pt}{1.6em}}\right.$};
		\node [anchor=west, style=none] (32) at (-0.75, 1) {$\vphantom{.}^{\otimes (p{-}1)}$};
		\node [xshift=-4pt, anchor=west, style=none] (33) at (-0.75, 0.5) {$\left.\vphantom{\rule{1pt}{1.6em}}\right)$};
		\node [xshift=4pt, anchor=east, style=none] (34) at (0.75, 0.5) {$\left(\vphantom{\rule{1pt}{1.6em}}\right.$};
		\node [xshift=-4pt, anchor=west, style=none] (35) at (1.25, 0.5) {$\left.\vphantom{\rule{1pt}{1.6em}}\right)$};
		\node [style=gn] (36) at (1, 0.25) {};
		\node [style=rn] (37) at (1, 0.75) {};
		\node [anchor=west, style=none] (38) at (1.25, 1) {$\vphantom{.}^{\otimes (p{-}1)}$};
	\end{pgfonlayer}
	\begin{pgfonlayer}{edgelayer}
		\draw [bend right, looseness=1.00] (3) to (4);
		\draw [bend left=15, looseness=1.00] (4) to (5);
		\draw [bend left, looseness=1.00] (0) to (4);
		\draw (4) to (2.center);
		\draw (7) to (8.center);
		\draw (10) to (13.center);
		\draw (12) to (11);
		\draw [bend right=75, looseness=1.25] (18) to (16);
		\draw (16) to (17.center);
		\draw [bend right, looseness=1.00] (18) to (16);
		\draw [bend right=75, looseness=1.25] (16) to (18);
		\draw (20) to (22.center);
		\draw [bend right=45, looseness=1.00] (29) to (30);
		\draw [bend right=45, looseness=1.00] (30) to (29);
		\draw (29) to (30);
		\draw [bend right=45, looseness=1.00] (37) to (36);
		\draw [bend right=45, looseness=1.00] (36) to (37);
		\draw (37) to (36);
	\end{pgfonlayer}
\end{tikzpicture}\]
\end{itemize}
Every rule but the $p$-supplementarity (with $p\in \mathbb{P}$ and $p\geq 3$) holds with this interpretation, so it cannot  be derived from the others:\\
$\forall p \in \mathbb{P},p\geq 3,~~ ZX_E\cup\{\suppc{n}\}_{n\in\mathbb{P}}\setminus\{\suppc{p}\}\nvdash \suppc{p}$
\end{proof}

\begin{corollary}
For any $n\geq3$ odd, the \frag{2n} of the $ZX_E$-Calculus is incomplete.
\end{corollary}

\begin{proof}
Let $p$ be an odd prime factor of $n$. Theorem \ref{th:supp-n-for-primes} proves that $ZX_E \nvdash (SUP_p)$, and notice that all the angles involved in the rule are multiples of $\frac{\pi}{2p}$, hence in the \frag{2n}.
\end{proof}

\begin{remark}
We can also notice that all the rules \suppn respect the quantity $\invR{.}$, so that the rule \e remains necessary.
\end{remark}

\subsection{Discussion on the Supplementarity's Derivability Structure}

Let $p$ and $q$ be two natural numbers. We have previously shown $ZX_E\cup\{\suppc{p},\suppc{q}\}\vdash \suppc{pq}$. 
In other words, \suppc{p} can be deduced from the supplementarity of the dividers of $p$. Now, can we deduce this same equality from the supplementarity of some of its multiples?

The first result comes when p is odd:

\begin{proposition}
\label{prop:sp-spq-to-sq}
\[\forall p,q\in \mathbb{N}^*,\qquad (p=1\mod 2) \implies \{\hlaw, \iv, \suppc{p}, \suppc{pq}\}\vdash \suppc{q}\]
\end{proposition}
\begin{proof}
In appendix at page \pageref{prf:sp-spq-to-sq}.
\end{proof}

There exists another -- weaker -- derivation when $p$ is even:
\begin{proposition}
\label{prop:sp-sp2q-to-spq}
\[\forall p,q\in \mathbb{N}^*,\qquad \{\hlaw, \iv, \suppc{p}, \suppc{p^2q}\}\vdash \suppc{pq}\]
\end{proposition}
\begin{proof}
In appendix at page \pageref{prf:sp-sp2q-to-spq}.
\end{proof}

\begin{remark}
In the last two propositions, we require that the ZX be ``general'' i.e.~with angles either real or a rational multiple of $\pi$ because we need $\alpha/p$ to be in the fragment in both cases. Though, the result can be expanded to any fragment for some $\alpha$ provided $\alpha/p$ be in the fragment.
\end{remark}

\noindent
To sum up:
\[ ZX_E\cup\{\suppc{p},\suppc{q}\}\vdash \suppc{pq} \]
\[ ZX_E\cup\{\suppc{p},\suppc{p^2q}\}\vdash \suppc{pq} \]
\[ (p=1\mod 2) \implies ZX_E\cup\{\suppc{p},\suppc{pq}\}\vdash \suppc{q} \]

\subsection{Updated Set of Rules}

We propose to add the generalisation of the supplementarity rule to the set of rules of the ZX-Calculus, and to restrict to the set necessary when dealing with particular fragments. We notice that the rule \ko is derivable from the others \cite{simplified-stabilizer} so we can get rid of it, and the new set of rules of the ZX-Calculus is shown in figure \ref{fig:ZX_rules3}.

\begin{figure}[!ht]
 \centering
  {\begin{tabular}{|ccccc|}
   \hline
   &&&& \\
   \begin{tikzpicture}[font={\footnotesize}]
	\begin{pgfonlayer}{nodelayer}
		\node [style=none] (0) at (-1.25, -0) {\rotatebox[origin=c]{63.43}{$~\cdots~$}};
		\node [style=none] (1) at (0.25, -0) {$=$};
		\node [style=gn] (2) at (1.5, 0) { \footnotesize$\alpha{+}\beta$};
		\node [style=gn, minimum width={0.5 cm}] (3) at (-0.7500001, -0.2500001) {\footnotesize $\beta$};
		\node [style=none] (4) at (-1.75, -0.5) {$~\cdots~$};
		\node [style=none] (5) at (2, -0.75) {};
		\node [style=none] (6) at (-1, -0.75) {};
		\node [style=none] (7) at (1.5, -0.75) {$~\cdots~$};
		\node [style=none] (8) at (-0.5, -0.75) {};
		\node [style=none] (9) at (1, -0.75) {};
		\node [style=none] (10) at (-2, -0.5) {};
		\node [style=none] (11) at (-1.5, -0.5) {};
		\node [style=none] (12) at (-0.75, -0.75) {$~\cdots~$};
		\node [style=none] (13) at (1.5, 0.75) {$~\cdots~$};
		\node [style=none] (14) at (1, 0.75) {};
		\node [style=none] (15) at (-2, 0.75) {};
		\node [style=none] (16) at (-0.5, 0.5) {};
		\node [style=none] (17) at (-1.5, 0.75) {};
		\node [style=none] (18) at (2, 0.75) {};
		\node [style=gn, minimum width={0.5 cm}] (19) at (-1.75, 0.25) {\footnotesize$\alpha$};
		\node [style=none] (20) at (-0.75, 0.5) {$~\cdots~$};
		\node [style=none] (21) at (-1, 0.5) {};
		\node [style=none] (22) at (-1.75, 0.75) {$~\cdots~$};
	\end{pgfonlayer}
	\begin{pgfonlayer}{edgelayer}
		\draw (3) to (16.center);
		\draw (3) to (6.center);
		\draw (3) to (8.center);
		\draw (19) to (10.center);
		\draw (19) to (11.center);
		\draw [bend right, looseness=1.00] (19) to (3);
		\draw [bend left, looseness=1.00] (19) to (3);
		\draw (14.center) to (2);
		\draw (2) to (9.center);
		\draw (5.center) to (2);
		\draw (2) to (18.center);
		\draw (19) to (15.center);
		\draw (19) to (17.center);
		\draw (3) to (21.center);
	\end{pgfonlayer}
\end{tikzpicture}&(S1) &$\quad$& \begin{tikzpicture}
	\begin{pgfonlayer}{nodelayer}
		\node [style=gn] (0) at (-0.7499998, -0) {};
		\node [style=none] (1) at (0, -0) {=};
		\node [style=none] (2) at (-0.7499998, 1) {};
		\node [style=none] (3) at (-0.7499998, -0.9999999) {};
		\node [style=none] (4) at (0.7499998, 1) {};
		\node [style=none] (5) at (0.7499998, -0.9999999) {};
	\end{pgfonlayer}
	\begin{pgfonlayer}{edgelayer}
		\draw (2) to (0);
		\draw (0) to (3);
		\draw (4) to (5);
	\end{pgfonlayer}
\end{tikzpicture}&(S2)\\
   &&&& \\
   \begin{tikzpicture}
	\begin{pgfonlayer}{nodelayer}
		\node [style=none] (0) at (0.7500001, -0.25) {};
		\node [style=none] (1) at (0, -0) {$=$};
		\node [style=gn] (2) at (1.25, 0.25) {};
		\node [style=none] (3) at (-0.7500001, -0.25) {};
		\node [style=none] (4) at (1.75, -0.25) {};
		\node [style=none] (5) at (-1.75, -0.25) {};
	\end{pgfonlayer}
	\begin{pgfonlayer}{edgelayer}
		\draw [in=90, out=90, looseness=1.75] (5.center) to (3.center);
		\draw [in=90, out=90, looseness=1.75] (0.center) to (4.center);
	\end{pgfonlayer}
\end{tikzpicture}&(S3) && \begin{tikzpicture}
	\begin{pgfonlayer}{nodelayer}
		\node [style=rn] (0) at (-0.7500001, -0.32) {$\frac{\text{-}\pi}{4}$};
		\node [style=gn] (1) at (-0.7499998, 0.32) {$\,\frac{\pi}{4}\,$};
		\node [style=none] (2) at (0, -0) {=};
		\node [style=none] (3) at (0.500001, 0.25) {};
		\node [style=none] (4) at (0.500001, -0.25) {};
		\node [style=none] (5) at (1, 0.25) {};
		\node [style=none] (6) at (1, -0.25) {};
	\end{pgfonlayer}
	\begin{pgfonlayer}{edgelayer}
		\draw (0) to (1);
		\draw [style=dashed] (3.center) to (5.center);
		\draw [style=dashed] (5.center) to (6.center);
		\draw [style=dashed] (6.center) to (4.center);
		\draw [style=dashed] (4.center) to (3.center);
	\end{pgfonlayer}
\end{tikzpicture}&(E)\\
   &&&& \\
   \begin{tikzpicture}
	\begin{pgfonlayer}{nodelayer}
		\node [style=gn] (0) at (0.75, 0) {};
		\node [style=none] (1) at (2.25, -0.25) {};
		\node [style=none] (2) at (0.5, -0.5) {};
		\node [style=rn] (3) at (2.25, 0.25) {};
		\node [style=none] (4) at (1, -0.5) {};
		\node [style=rn] (5) at (0.75, 0.5) {};
		\node [style=rn] (6) at (2.75, 0.25) {};
		\node [style=none] (7) at (2.75, -0.25) {};
		\node [style=none] (8) at (1.5, 0) {$=$};
		\node [style=rn] (9) at (0, 0.25) {};
		\node [style=gn] (10) at (0, -0.25) {};
	\end{pgfonlayer}
	\begin{pgfonlayer}{edgelayer}
		\draw [style=none] (5) to (0);
		\draw[bend right=23]  [style=none] (0) to (2.center);
		\draw[bend left=23]  [style=none] (0) to (4.center);
		\draw [style=none] (3) to (1.center);
		\draw [style=none] (6) to (7.center);
		\draw (9) to (10);
	\end{pgfonlayer}
\end{tikzpicture}&(B1) && \begin{tikzpicture}
	\begin{pgfonlayer}{nodelayer}
		\node [style=none] (0) at (3.75, 0.75) {};
		\node [style=rn] (1) at (0.5, -0.25) {};
		\node [style=none] (2) at (3.25, -0.75) {};
		\node [style=none] (3) at (1.25, 1) {};
		\node [style=none] (4) at (3.25, 0.75) {};
		\node [style=none] (5) at (0.5, -0.75) {};
		\node [style=none] (6) at (0.5, 1) {};
		\node [style=gn] (7) at (1.25, 0.5) {};
		\node [style=none] (8) at (2.25, 0) {$=$};
		\node [style=rn] (9) at (1.25, -0.25) {};
		\node [style=gn] (10) at (3.5, -0.25) {};
		\node [style=gn] (11) at (0.5, 0.5) {};
		\node [style=none] (12) at (1.25, -0.75) {};
		\node [style=none] (13) at (3.75, -0.75) {};
		\node [style=rn] (14) at (3.5, 0.25) {};
		\node [style=rn] (15) at (0, 0.25) {};
		\node [style=gn] (16) at (0, -0.25) {};
	\end{pgfonlayer}
	\begin{pgfonlayer}{edgelayer}
		\draw [style=none] (12.center) to (9);
		\draw [style=none] (5.center) to (1);
		\draw [style=none] (7) to (3.center);
		\draw [style=none, bend right=23, looseness=1.00] (9) to (7);
		\draw [style=none] (11) to (6.center);
		\draw [style=none, bend left=23, looseness=1.00] (1) to (11);
		\draw [style=none,bend right=23] (13.center) to (10);
		\draw [style=none] (10) to (14);
		\draw [style=none,bend left=23] (14) to (4.center);
		\draw [style=none,bend right=23] (14) to (0.center);
		\draw[bend right=23] (10) to (2.center);
		\draw (11) to (9);
		\draw (7) to (1);
		\draw (15) to (16);
	\end{pgfonlayer}
\end{tikzpicture}&(B2)\\
   &&&& \\
   \begin{tikzpicture}
	\begin{pgfonlayer}{nodelayer}
		\node [style=none] (0) at (0, -0) {$=$};
		\node [style=none] (1) at (-0.7499998, 0.9999999) {};
		\node [style=none] (2) at (-0.7499998, -0.7499998) {};
		\node [style=none] (3) at (1.5, 0.9999999) {};
		\node [style=gn] (4) at (-0.7499998, -0.2500001) {$~\pi~$};
		\node [style=none] (5) at (1.5, -0.7499998) {};
		\node [style=rn] (6) at (-0.7499998, 0.5) {$~\alpha~$};
		\node [style=rn] (7) at (1.5, -0.2500001) {$-\alpha$};
		\node [style=gn] (8) at (1.5, 0.5) {$~\pi~$};
		\node [style=rn] (9) at (-1.5, 0.2500001) {};
		\node [style=gn] (10) at (-1.5, -0.2500001) {};
		\node [style=rn] (11) at (0.7499998, 0.5) {$~\alpha~$};
		\node [style=gn] (12) at (0.7499998, -0.2500001) {$~\pi~$};
	\end{pgfonlayer}
	\begin{pgfonlayer}{edgelayer}
		\draw (3.center) to (8);
		\draw (8) to (7);
		\draw (7) to (5.center);
		\draw (4) to (2.center);
		\draw (1.center) to (6);
		\draw (6) to (4);
		\draw (9) to (10);
		\draw (12) to (11);
	\end{pgfonlayer}
\end{tikzpicture}&(K) &&  \begin{tikzpicture}
	\begin{pgfonlayer}{nodelayer}
		\node [style=none] (0) at (0.7500001, 1) {};
		\node [style=none] (1) at (-0.7500001, -1) {};
		\node [style={H box}] (2) at (-0.7500001, 0.5000001) {};
		\node [style=rn] (3) at (-1.25, -0) {\footnotesize$~\alpha~$};
		\node [style={H box}] (4) at (-1.75, -0.5000001) {};
		\node [style=none] (5) at (-0.7500001, 1) {};
		\node [style={H box}] (6) at (-0.7500001, -0.5000001) {};
		\node [style=none] (7) at (-1.25, 0.7500001) {$\cdots$};
		\node [style=none] (8) at (0, -0) {$=$};
		\node [style=gn] (9) at (1.25, -0) {\footnotesize$~\alpha~$};
		\node [style=none] (10) at (1.75, -1) {};
		\node [style=none] (11) at (1.75, 1) {};
		\node [style=none] (12) at (-1.75, 1) {};
		\node [style={H box}] (13) at (-1.75, 0.5000001) {};
		\node [style=none] (14) at (-1.75, -1) {};
		\node [style=none] (15) at (0.7500001, -1) {};
		\node [style=none] (16) at (-1.25, -0.7500001) {$\cdots$};
		\node [style=none] (17) at (1.25, 0.7500001) {$\cdots$};
		\node [style=none] (18) at (1.25, -0.7500001) {$\cdots$};
	\end{pgfonlayer}
	\begin{pgfonlayer}{edgelayer}
		\draw [bend right, looseness=1.00] (3) to (4);
		\draw [bend left, looseness=1.00] (3) to (6);
		\draw (6) to (1.center);
		\draw (4) to (14.center);
		\draw [bend left, looseness=1.00] (3) to (13);
		\draw [bend right, looseness=1.00] (3) to (2);
		\draw (2) to (5.center);
		\draw (13) to (12.center);
		\draw [bend left=23, looseness=1.00] (9) to (0.center);
		\draw [bend right=23, looseness=1.00] (9) to (11.center);
		\draw [bend right=23, looseness=1.00] (9) to (15.center);
		\draw [bend left=23, looseness=1.00] (9) to (10.center);
	\end{pgfonlayer}
\end{tikzpicture}&(H)\\
   &&&& \\
   \multicolumn{5}{|c|}{\begin{tabular}{ccccc}
   \begin{tikzpicture}
	\begin{pgfonlayer}{nodelayer}
		\node [style=gn] (0) at (0.5, -0.55) {$~\frac{\pi}{2}~$};
		\node [style=rn] (1) at (0.5, -0) {};
		\node [style=gn] (2) at (0.5, 0.55) {$~\frac{\pi}{2}~$};
		\node [style=none] (3) at (0.5, 0.9999999) {};
		\node [style=gn] (4) at (1.25, 0.5) {$\frac{-\pi}{2}$};
		\node [style=none] (5) at (0.5, -0.9999999) {};
		\node [style=none] (6) at (-0.9999999, 0.9999999) {};
		\node [style=none] (7) at (-0.9999999, -0.9999999) {};
		\node [style=none] (8) at (-0.2500001, -0) {$=$};
		\node [style={{H box}}] (9) at (-0.9999999, -0) {};
	\end{pgfonlayer}
	\begin{pgfonlayer}{edgelayer}
		\draw (3.center) to (2);
		\draw (2) to (1);
		\draw (1) to (0);
		\draw (0) to (5.center);
		\draw (1) to (4);
		\draw (6.center) to (7.center);
	\end{pgfonlayer}
\end{tikzpicture}&(EU) &$\quad$& \begin{tikzpicture}
	\begin{pgfonlayer}{nodelayer}
		\node [style=rn] (0) at (-1.75, -0) {};
		\node [style=none] (1) at (0.5, -0) {=};
		\node [style=none] (2) at (-1.75, -0.75) {};
		\node [style=gn] (3) at (-2, 1) {\scriptsize $\alpha{+}\frac{2\pi}{n}$};
		\node [style=gn, align=center] (4) at (-0.25, 1) {\scriptsize $\alpha{+}$\\ \scriptsize $\frac{n-1}{n}2\pi$};
		\node [style=gn] (5) at (-3, 1) {$~\alpha~$};
		\node [style=none] (6) at (-1.25, 0.5) {$\cdots$};
		\node [style=none] (7) at (1.75, -0.75) {};
		\node [style=rn] (8) at (1.75, -0) {};
		\node [style=gn, align=center] (9) at (1.75, 1.25) {\scriptsize $n\alpha$+\\\scriptsize $(n{-}1)\pi$};
		\node [style=none] (10) at (1.875, 0.5) {$\cdots$};
	\end{pgfonlayer}
	\begin{pgfonlayer}{edgelayer}
		\draw [bend right, looseness=1.00] (5) to (0);
		\draw [bend left=15, looseness=0.75] (0) to (3);
		\draw [bend left, looseness=1.00] (4) to (0);
		\draw (0) to (2.center);
		\draw [bend right=75, looseness=1.25] (9) to (8);
		\draw (8) to (7.center);
		\draw [bend right, looseness=1.00] (9) to (8);
		\draw [bend right=75, looseness=1.25] (8) to (9);
	\end{pgfonlayer}
\end{tikzpicture}&(SUP$_n$) \\
   &&& $n\in \mathbb{N}^*$ or $n\in \mathbb{P}$ &
   \end{tabular}} \\
   \hline
  \end{tabular}}
 \caption[]{New set of rules for the ZX-calculus with scalars. All of these rules also hold when flipped upside-down, or with the colours red and green swapped. The right-hand side of (E) is an empty diagram. ($\cdots$) in (S1) and (H) denote $0$ or more wires, while (\protect\rotatebox{45}{\raisebox{-0.4em}{$\cdots$}}) denote $1$ or more wires.}
 \label{fig:ZX_rules3}
\end{figure}
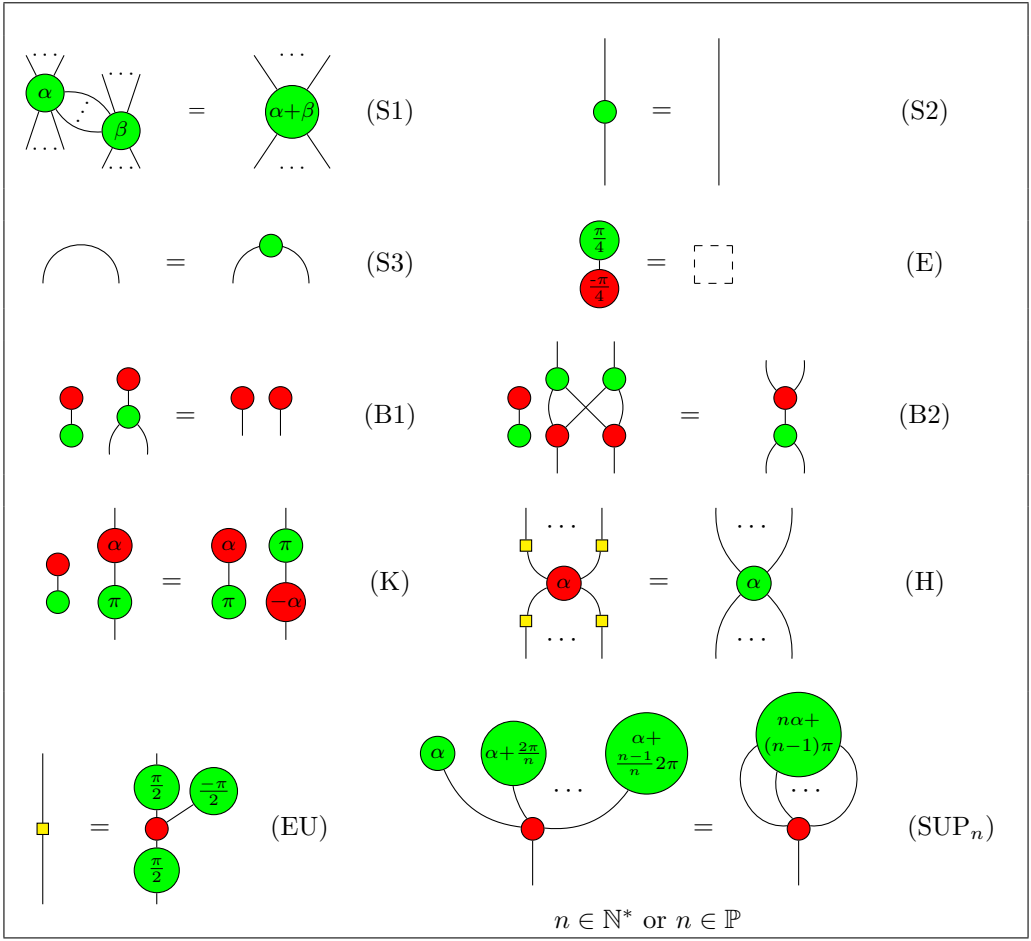

\begin{remark}
We can prove that \suppt is not derivable from $ZX_E\setminus\{\suppt\}$, using the interpretation $\interp{.}^{\sharp}_{k,l}$ defined in \cite{supplementarity} with $k=3$ and $l=8$.

However, it is important to notice that we have not proven that \suppt can not be derived from the rest once the cyclotomic supplementarity is added. Indeed, the family of interpretations used in the proof of Theorem \ref{th:supp-n-for-primes} only works when $p$ is odd, and the one used previously, $\interp{.}^{\sharp}_{k,l}$, does not hold for many supplementarity rules.

The rule \suppt is all the more peculiar as, due to the Hopf law \hlaw, it is the only supplementarity rule that creates a non-trivial scalar -- except for the supplementarity rules for even numbers, which anyway derive from \suppt. For instance it may create a scalar worth $0$ -- when applied with $\alpha=0$ -- which is the first step towards proving the rule \zo in the \frag{4}. Moreover, it is the only supplementarity that can not be proven to be necessary by simply multiplying the angles by a constant.
\end{remark}

\subsection{The General ZX-Calculus is still Not Complete}

The argument given by Schr\"oder de Witt and Zamdzhiev \cite{incompleteness} to show the incompleteness of the general ZX-Calculus is not valid anymore -- when multiplying the angles by any integer, there is at least one supplementarity that does not hold. But we can patch the demonstration to make it valid again.

\begin{theorem}
The general ZX-Calculus is incomplete with the set of rules in figure \ref{fig:ZX_rules3}.
\end{theorem}

\begin{proof}
We will make the proof using a combination of ZX-rules and matrix calculus on the interpretations of the diagrams. Consider the following diagrams:
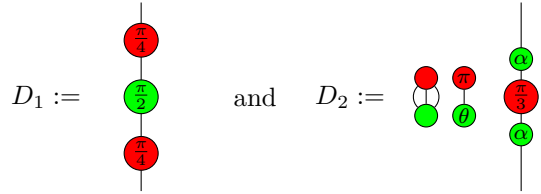
\[\begin{tikzpicture}
	\begin{pgfonlayer}{nodelayer}
		\node [style=none] (0) at (-7.25, -0) {$D_1 :=$};
		\node [style=none] (1) at (-6, 1.25) {};
		\node [style=none] (2) at (-6, -1.25) {};
		\node [style=none] (3) at (-4.5, -0) {and};
		\node [style=gn] (4) at (-6, -0) {$\frac{\pi}{2}$};
		\node [style=rn] (5) at (-6, 0.7500001) {$\frac{\pi}{4}$};
		\node [style=rn] (6) at (-6, -0.7500001) {$\frac{\pi}{4}$};
		\node [style=gn] (7) at (-1, 0.5000001) {$\alpha$};
		\node [style=none] (8) at (-3.25, -0) {$D_2 :=$};
		\node [style=gn] (9) at (-1, -0.5000001) {$\alpha$};
		\node [style=none] (10) at (-1, 1.25) {};
		\node [style=rn] (11) at (-1, -0) {$\frac{\pi}{3}$};
		\node [style=none] (12) at (-1, -1.25) {};
		\node [style=rn] (13) at (-2.25, 0.25) {};
		\node [style=gn] (14) at (-2.25, -0.25) {};
		\node [style=rn] (15) at (-1.75, 0.2500001) {$\pi$};
		\node [style=gn] (16) at (-1.75, -0.2500001) {$\theta$};
	\end{pgfonlayer}
	\begin{pgfonlayer}{edgelayer}
		\draw (1.center) to (2.center);
		\draw (10.center) to (12.center);
		\draw [bend right=45, looseness=1.00] (13) to (14);
		\draw (15) to (16);
		\draw [bend left=45, looseness=1.00] (13) to (14);
		\draw (14) to (13);
	\end{pgfonlayer}
\end{tikzpicture}\]
We will try to express $D_1$ in the form $D_2$. One can notice that $\interp{D_1}=\interp{D_2}$ when $\alpha_0=\frac{\pi}{2}-\arccos\left(\sqrt{\frac{2}{3}}\right)$ and $\theta_0=\arccos\left(\frac{\sqrt{2}}{2}+\frac{\sqrt{3}}{6}\right)$. We can even show:

\textbf{No more than four values for $\alpha$ are possible when decomposing $D_1$ in the form $D_2$:}\\
When applying the $\pi$ green state at the top and the $0$ green state at the bottom of both $D_1$ and $D_2$, we end up with:
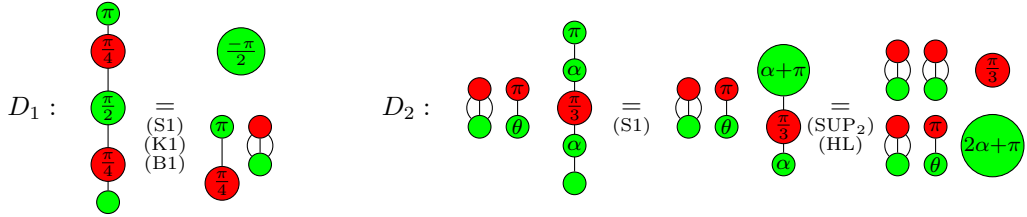
\[D_1:\quad\begin{tikzpicture}
	\begin{pgfonlayer}{nodelayer}
		\node [style=gn] (0) at (-5.75, -0) {$\frac{\pi}{2}$};
		\node [style=rn] (1) at (-5.75, 0.7499999) {$\frac{\pi}{4}$};
		\node [style=rn] (2) at (-5.75, -0.7500001) {$\frac{\pi}{4}$};
		\node [style=gn] (3) at (-5.75, 1.25) {$\pi$};
		\node [style=gn] (4) at (-5.75, -1.25) {};
		\node [style=none] (5) at (-5, -0) {=};
		\node [style=gn] (6) at (-4.25, -0.2500001) {$\pi$};
		\node [style=rn] (7) at (-4.25, -1) {$\frac{\pi}{4}$};
		\node [style=gn] (8) at (-4, 0.7499999) {$\frac{-\pi}{2}$};
		\node [style=rn] (9) at (-3.75, -0.2500001) {};
		\node [style=gn] (10) at (-3.75, -0.75) {};
		\node [style=none] (11) at (-5, -0.75) {\scriptsize \bo};
		\node [style=none] (12) at (-5, -0.25) {\scriptsize \sone};
		\node [style=none] (13) at (-5, -0.5) {\scriptsize \ko};
	\end{pgfonlayer}
	\begin{pgfonlayer}{edgelayer}
		\draw (3) to (4);
		\draw (6) to (7);
		\draw [bend right=45, looseness=1.00] (9) to (10);
		\draw [bend left=45, looseness=1.00] (9) to (10);
		\draw (10) to (9);
	\end{pgfonlayer}
\end{tikzpicture}\qquad\qquad D_2:\quad\begin{tikzpicture}
	\begin{pgfonlayer}{nodelayer}
		\node [style=gn] (0) at (-5.5, 0.5000001) {$\alpha$};
		\node [style=gn] (1) at (-5.5, -0.5000001) {$\alpha$};
		\node [style=rn] (2) at (-5.5, -0) {$\frac{\pi}{3}$};
		\node [style=rn] (3) at (-6.75, 0.25) {};
		\node [style=gn] (4) at (-6.75, -0.25) {};
		\node [style=rn] (5) at (-6.25, 0.25) {$\pi$};
		\node [style=gn] (6) at (-6.25, -0.25) {$\theta$};
		\node [style=gn] (7) at (-5.5, 1) {$\pi$};
		\node [style=gn] (8) at (-5.5, -1) {};
		\node [style=none] (9) at (-4.75, -0) {=};
		\node [style=gn] (10) at (-2.75, -0.7500001) {$\alpha$};
		\node [style=rn] (11) at (-3.5, 0.25) {$\pi$};
		\node [style=rn] (12) at (-2.75, -0.25) {$\frac{\pi}{3}$};
		\node [style=rn] (13) at (-4, 0.25) {};
		\node [style=gn] (14) at (-3.5, -0.25) {$\theta$};
		\node [style=gn] (15) at (-4, -0.25) {};
		\node [style=gn] (16) at (-2.75, 0.5000001) {$\alpha{+}\pi$};
		\node [style=rn] (17) at (-1.25, -0.2500001) {};
		\node [style=gn] (18) at (0, -0.5) {$2\alpha{+}\pi$};
		\node [style=gn] (19) at (-0.7499998, -0.75) {$\theta$};
		\node [style=rn] (20) at (0, 0.5) {$\frac{\pi}{3}$};
		\node [style=rn] (21) at (-0.7499998, -0.2500001) {$\pi$};
		\node [style=gn] (22) at (-1.25, -0.75) {};
		\node [style=gn] (23) at (-1.25, 0.2500001) {};
		\node [style=rn] (24) at (-1.25, 0.75) {};
		\node [style=gn] (25) at (-0.7499998, 0.2500001) {};
		\node [style=rn] (26) at (-0.7499998, 0.75) {};
		\node [style=none] (27) at (-2, -0) {=};
		\node [style=none] (28) at (-4.75, -0.25) {\scriptsize \sone};
		\node [style=none] (29) at (-2, -0.25) {\scriptsize \suppt};
		\node [style=none] (30) at (-2, -0.4999999) {\scriptsize \hlaw};
	\end{pgfonlayer}
	\begin{pgfonlayer}{edgelayer}
		\draw [bend right=45, looseness=1.00] (3) to (4);
		\draw (5) to (6);
		\draw [bend left=45, looseness=1.00] (3) to (4);
		\draw (4) to (3);
		\draw (7) to (8);
		\draw [bend right=45, looseness=1.00] (13) to (15);
		\draw (11) to (14);
		\draw [bend left=45, looseness=1.00] (13) to (15);
		\draw (15) to (13);
		\draw (16) to (10);
		\draw [bend right=45, looseness=1.00] (17) to (22);
		\draw (21) to (19);
		\draw [bend left=45, looseness=1.00] (17) to (22);
		\draw (22) to (17);
		\draw [bend right=45, looseness=1.00] (24) to (23);
		\draw [bend left=45, looseness=1.00] (24) to (23);
		\draw (23) to (24);
		\draw [bend right=45, looseness=1.00] (26) to (25);
		\draw [bend left=45, looseness=1.00] (26) to (25);
		\draw (25) to (26);
	\end{pgfonlayer}
\end{tikzpicture}\]
In order for their interpretations to be equal, we need:
\[ e^{i\frac{\pi}{4}}\sqrt{2}e^{-i\frac{\pi}{4}}=\frac{1}{2}e^{i\theta}(1+e^{i\frac{\pi}{3}})(1+e^{i(2\alpha+\pi)}) \quad\text{i.e.}\quad \frac{\sqrt{2}}{2} = e^{i(\theta+\frac{\pi}{6}+\alpha+\frac{\pi}{2})}\cos\left(\frac{\pi}{6}\right)\cos\left(\alpha+\frac{\pi}{2}\right)\]
So using the modulus, $\abs{\cos\left(\alpha+\frac{\pi}{2}\right)}=\sqrt{\frac{2}{3}}$, thus $\alpha=\pm\frac{\pi}{2}\pm\arccos\left(\sqrt{\frac{2}{3}}\right)\mod 2\pi$.

\textbf{$\alpha$ is not a rational multiple of $\pi$:}\\
One can check that $e^{i\alpha_0}$ is a root of the polynomial $3X^4+2X^2+3$ which is irreducible in $\mathbb{Z}$ (since $30203$ is a prime number, thanks to Cohn's irreducibility criterion, $3X^4+2X^2+3$ is irreducible in $\mathbb{Z}$). The polynomial is not cyclotomic because its coefficient of higher degree is not $1$, hence $e^{i\alpha_0}$ is not a root of unity, i.e.~$\alpha_0$ is not a rational multiple of $\pi$. As a consequence, none of the $\pm\frac{\pi}{2}\pm\arccos\left(\sqrt{\frac{2}{3}}\right)$ are rational multiples of $\pi$.

Now, let us put back all the pieces together. Assume $ZX\vdash D_1=D_2$ for some $\alpha$ and $\theta$. Then there exists a finite sequence of rules of the ZX that transforms $D_2$ into $D_1$. We define $q\in\mathbb{N}^*$ such that for any \suppc{p} in the sequence, $p\leq q$, and 
$S=\{k(q+4)!+1~|~k\in\mathbb{N}\}$.\\
For all $q'\in S$ and for $\interpspe{.}_{q'}$ the interpretation that multiplies the angles by $q'$, the rules of the ZX are preserved, $\interpspe{D_1}_{q'}=D_1$,  and $\interpspe{D_2}_{q'}$ is in the form $D_2$. Indeed:
\begin{itemize}
\item $(q+4)!$ is clearly a multiple of $8$ so $q'\frac{\pi}{4}=\frac{\pi}{4} \mod 2\pi$, so all the rules but the supplementarity rules hold, and $\interpspe{D_1}_{q'}=D_1$ since it is in the \frag{4}.
\item for any $p\in \mathbb{P}$ such that $p\leq q$, then $(q+4)! = 0 \mod p$ so $gcd(p, q') =1$, which implies that \suppc{p} also holds.
\item $(q+4)!$ is a multiple of $6$ so $q'\frac{\pi}{3}=\frac{\pi}{3} \mod 2\pi$ hence $\interpspe{D_2}_{q'}$ is in the form $D_2$.
\end{itemize}
Then, 
$ZX\vdash D_1=\interpspe{D_2}_{q'}$.\\
$D_1$ has a finite number of decompositions in the form $D_2$, but $\left\lbrace\interpspe{D_2}_{q'}~|~q'\in S\right\rbrace$ is infinite -- since $\alpha$ is an irrational multiple of $\pi$ -- and all these diagrams are decompositions of $D_1$ in the form $D_2$, hence we end up with a contradiction.\\
So $ZX\nvdash D_1=D_2$, which proves the incompleteness.
\end{proof}

%
%
%
%
\bibliography{Cyclotomic-Supplementarity}

\appendix
\section{Appendix}

\subsection{(ZO) is Derivable in ZX$_{\boldsymbol{E}}$}

The following scalar lemmas are provable in ZX \cite{scalar-completeness,simplified-stabilizer}:\\
\begin{minipage}[t]{0.5\textwidth}
\begin{lemma}
\label{lem:2isroot2square}
\[ \begin{tikzpicture}
	\begin{pgfonlayer}{nodelayer}
		\node [style=gn] (0) at (0.7500001, 0.25) {};
		\node [style=gn] (1) at (0.25, 0.25) {};
		\node [style=rn] (2) at (0.25, -0.25) {};
		\node [style=rn] (3) at (0.7500001, -0.25) {};
		\node [style=none] (4) at (1.5, 0) {$=$};
		\node [style=gn] (5) at (2.25, -0) {};
	\end{pgfonlayer}
	\begin{pgfonlayer}{edgelayer}
		\draw (2) to (1);
		\draw (3) to (0);
	\end{pgfonlayer}
\end{tikzpicture}\]
\end{lemma}
\end{minipage}
\begin{minipage}[t]{0.5\textwidth}
\begin{lemma}
\label{lem:bicolor-alpha-0}
\[ \begin{tikzpicture}
	\begin{pgfonlayer}{nodelayer}
		\node [style=gn] (0) at (-0.7500001, 0.25) {};
		\node [style=rn] (1) at (-0.7500001, -0.25) {$\alpha$};
		\node [style=none] (2) at (0, -0) {=};
		\node [style=rn] (3) at (0.7500001, -0.25) {};
		\node [style=gn] (4) at (0.7500001, 0.25) {};
	\end{pgfonlayer}
	\begin{pgfonlayer}{edgelayer}
		\draw (1) to (0);
		\draw (3) to (4);
	\end{pgfonlayer}
\end{tikzpicture}\]
\end{lemma}
\end{minipage}

\begin{proof}[Proof of Proposition \ref{prop:zo-deducible}]
\phantomsection\label{prf:zo-deducible}
First, using \stwo, \hlaw, \supp, \bo and \sone, and then applying a red state at the top:
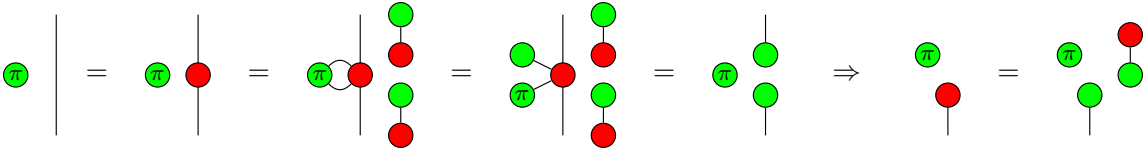
\[ \fit{\begin{tikzpicture}
	\begin{pgfonlayer}{nodelayer}
		\node [style=none] (0) at (-6.5, -0.7500001) {};
		\node [style=gn] (1) at (-7, -0) {$\pi$};
		\node [style=none] (2) at (-6, -0) {=};
		\node [style=none] (3) at (-6.5, 0.7500001) {};
		\node [style=none] (4) at (-4.75, -0.7500001) {};
		\node [style=gn] (5) at (-5.25, -0) {$\pi$};
		\node [style=none] (6) at (-4.75, 0.7500001) {};
		\node [style=rn] (7) at (-4.75, -0) {};
		\node [style=none] (8) at (-2.75, -0.7500001) {};
		\node [style=gn] (9) at (-3.25, -0) {$\pi$};
		\node [style=none] (10) at (-2.75, 0.7500001) {};
		\node [style=none] (11) at (-4, -0) {=};
		\node [style=rn] (12) at (-2.75, -0) {};
		\node [style=rn] (13) at (-2.25, -0.75) {};
		\node [style=gn] (14) at (-2.25, -0.25) {};
		\node [style=none] (15) at (-1.5, -0) {=};
		\node [style=none] (16) at (2.25, 0.7499999) {};
		\node [style=gn] (17) at (1.75, -0) {$\pi$};
		\node [style=gn] (18) at (2.25, -0.2500001) {};
		\node [style=none] (19) at (2.25, -0.7499999) {};
		\node [style=gn] (20) at (2.25, 0.2500001) {};
		\node [style=gn] (21) at (-2.25, 0.75) {};
		\node [style=rn] (22) at (-2.25, 0.2500001) {};
		\node [style=gn] (23) at (0.2499997, 0.7499999) {};
		\node [style=rn] (24) at (0.2499997, -0.7499999) {};
		\node [style=rn] (25) at (-0.2499997, -0) {};
		\node [style=rn] (26) at (0.2499997, 0.2500001) {};
		\node [style=none] (27) at (-0.2499997, -0.7499999) {};
		\node [style=none] (28) at (-0.2499997, 0.7499999) {};
		\node [style=gn] (29) at (0.2499997, -0.2500001) {};
		\node [style=gn] (30) at (-0.7499999, -0.2500001) {$\pi$};
		\node [style=gn] (31) at (-0.7499999, 0.2500001) {};
		\node [style=none] (32) at (0.9999996, -0) {=};
		\node [style=none] (33) at (3.25, -0) {$\implies$};
		\node [style=none] (34) at (6.25, -0.7500001) {};
		\node [style=gn] (35) at (6, 0.25) {$\pi$};
		\node [style=gn] (36) at (6.25, -0.25) {};
		\node [style=rn] (37) at (6.75, 0.5000001) {};
		\node [style=gn] (38) at (6.75, -0) {};
		\node [style=none] (39) at (5.25, -0) {=};
		\node [style=gn] (40) at (4.25, 0.25) {$\pi$};
		\node [style=rn] (41) at (4.5, -0.25) {};
		\node [style=none] (42) at (4.5, -0.7500001) {};
	\end{pgfonlayer}
	\begin{pgfonlayer}{edgelayer}
		\draw (3.center) to (0.center);
		\draw (6.center) to (4.center);
		\draw (10.center) to (8.center);
		\draw (14) to (13);
		\draw (16.center) to (20);
		\draw (18) to (19.center);
		\draw [bend right=45, looseness=1.25] (9) to (12);
		\draw [bend left=45, looseness=1.25] (9) to (12);
		\draw (21) to (22);
		\draw (28.center) to (27.center);
		\draw (29) to (24);
		\draw (23) to (26);
		\draw (31) to (25);
		\draw (25) to (30);
		\draw (36) to (34.center);
		\draw (38) to (37);
		\draw (41) to (42.center);
	\end{pgfonlayer}
\end{tikzpicture}}\]
Then, using \iv, \sone, the previous result and lemma \ref{lem:bicolor-alpha-0}:
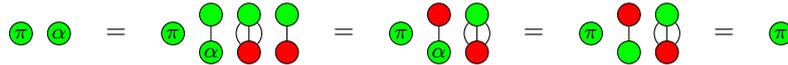
\[ \begin{tikzpicture}
	\begin{pgfonlayer}{nodelayer}
		\node [style=none] (0) at (-3.75, -0) {=};
		\node [style=gn] (1) at (-3, -0) {$\pi$};
		\node [style=rn] (2) at (-2, -0.25) {};
		\node [style=gn] (3) at (-2, 0.25) {};
		\node [style=gn] (4) at (-5, -0) {$\pi$};
		\node [style=gn] (5) at (-2.5, -0.25) {$\alpha$};
		\node [style=gn] (6) at (-1.5, 0.25) {};
		\node [style=rn] (7) at (-1.5, -0.25) {};
		\node [style=none] (8) at (-0.75, -0) {=};
		\node [style=gn] (9) at (-4.5, -0) {$\alpha$};
		\node [style=gn] (10) at (-2.5, 0.25) {};
		\node [style=gn] (11) at (1, 0.25) {};
		\node [style=gn] (12) at (0.5, -0.25) {$\alpha$};
		\node [style=rn] (13) at (1, -0.25) {};
		\node [style=gn] (14) at (0, -0) {$\pi$};
		\node [style=rn] (15) at (0.5, 0.25) {};
		\node [style=gn] (16) at (3, -0.25) {};
		\node [style=gn] (17) at (2.5, -0) {$\pi$};
		\node [style=none] (18) at (1.75, -0) {=};
		\node [style=rn] (19) at (3.5, -0.25) {};
		\node [style=gn] (20) at (3.5, 0.25) {};
		\node [style=rn] (21) at (3, 0.25) {};
		\node [style=none] (22) at (4.25, -0) {=};
		\node [style=gn] (23) at (5, -0) {$\pi$};
	\end{pgfonlayer}
	\begin{pgfonlayer}{edgelayer}
		\draw [bend left=45, looseness=1.00] (3) to (2);
		\draw [bend right=45, looseness=1.00] (3) to (2);
		\draw (2) to (3);
		\draw (7) to (6);
		\draw (10) to (5);
		\draw [bend left=45, looseness=1.00] (11) to (13);
		\draw [bend right=45, looseness=1.00] (11) to (13);
		\draw (13) to (11);
		\draw (15) to (12);
		\draw [bend left=45, looseness=1.00] (20) to (19);
		\draw [bend right=45, looseness=1.00] (20) to (19);
		\draw (19) to (20);
		\draw (21) to (16);
	\end{pgfonlayer}
\end{tikzpicture}\]
Notice that so far we have not used \e. It becomes necessary in the following, since $\invR{.}$ does not hold there. Using \e, the previous results, and lemmas \ref{lem:bicolor-alpha-0} and \ref{lem:2isroot2square}:
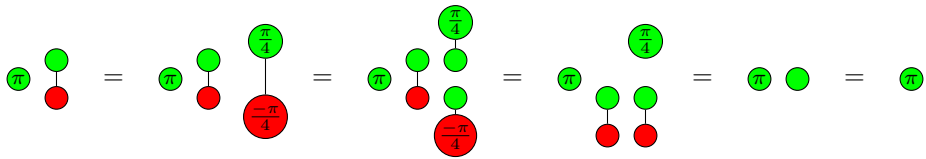
\[ {\begin{tikzpicture}
	\begin{pgfonlayer}{nodelayer}
		\node [style=gn] (0) at (-7, -0) {$\pi$};
		\node [style=gn] (1) at (-6.5, 0.25) {};
		\node [style=rn] (2) at (-6.5, -0.25) {};
		\node [style=none] (3) at (-5.75, -0) {=};
		\node [style=rn, inner sep=0] (4) at (-3.75, -0.5000001) {$\frac{-\pi}{4}$};
		\node [style=gn] (5) at (-3.75, 0.5000001) {$\frac{\pi}{4}$};
		\node [style=rn] (6) at (-4.5, -0.25) {};
		\node [style=gn] (7) at (-5, -0) {$\pi$};
		\node [style=gn] (8) at (-4.5, 0.25) {};
		\node [style=gn] (9) at (-1.25, 0.7500001) {$\frac{\pi}{4}$};
		\node [style=gn] (10) at (-1.75, 0.25) {};
		\node [style=none] (11) at (-3, -0) {=};
		\node [style=gn] (12) at (-2.25, -0) {$\pi$};
		\node [style=rn] (13) at (-1.75, -0.25) {};
		\node [style=gn] (14) at (-1.25, 0.25) {};
		\node [style=gn] (15) at (-1.25, -0.25) {};
		\node [style=gn] (16) at (0.25, -0) {$\pi$};
		\node [style=rn] (17) at (0.7500001, -0.7500001) {};
		\node [style=gn] (18) at (1.25, -0.25) {};
		\node [style=gn] (19) at (1.25, 0.5000001) {$\frac{\pi}{4}$};
		\node [style=rn] (20) at (1.25, -0.7500001) {};
		\node [style=gn] (21) at (0.7500001, -0.25) {};
		\node [style=none] (22) at (-0.5000001, -0) {=};
		\node [style=none] (23) at (2, -0) {=};
		\node [style=gn] (24) at (4.75, -0) {$\pi$};
		\node [style=rn, inner sep={-0.2 pt}] (25) at (-1.25, -0.7500001) {$\frac{-\pi}{4}$};
		\node [style=gn] (26) at (3.25, -0) {};
		\node [style=gn] (27) at (2.75, -0) {$\pi$};
		\node [style=none] (28) at (4, -0) {=};
	\end{pgfonlayer}
	\begin{pgfonlayer}{edgelayer}
		\draw (2) to (1);
		\draw (4) to (5);
		\draw (6) to (8);
		\draw (13) to (10);
		\draw (9) to (14);
		\draw (17) to (21);
		\draw (18) to (20);
		\draw (15) to (25);
	\end{pgfonlayer}
\end{tikzpicture}}\]
Finally:\\\vspace{-0.5em}\\
\indent$ 
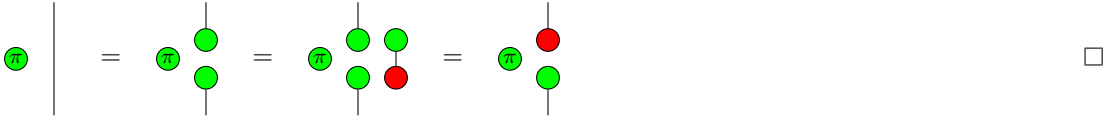
\begin{tikzpicture}
	\begin{pgfonlayer}{nodelayer}
		\node [style=none] (0) at (-3, -0.75) {};
		\node [style=gn] (1) at (-3.5, -0) {$\pi$};
		\node [style=none] (2) at (-2.25, -0) {=};
		\node [style=none] (3) at (-3, 0.75) {};
		\node [style=none] (4) at (-0.25, -0) {=};
		\node [style=rn] (5) at (1.5, -0.25) {};
		\node [style=gn] (6) at (1.5, 0.25) {};
		\node [style=none] (7) at (2.25, -0) {=};
		\node [style=none] (8) at (-1, 0.75) {};
		\node [style=gn] (9) at (-1.5, -0) {$\pi$};
		\node [style=gn] (10) at (-1, -0.25) {};
		\node [style=none] (11) at (-1, -0.75) {};
		\node [style=gn] (12) at (-1, 0.25) {};
		\node [style=gn] (13) at (1, 0.25) {};
		\node [style=gn] (14) at (0.5, -0) {$\pi$};
		\node [style=none] (15) at (1, -0.75) {};
		\node [style=none] (16) at (1, 0.75) {};
		\node [style=gn] (17) at (1, -0.25) {};
		\node [style=none] (18) at (3.5, -0.75) {};
		\node [style=gn] (19) at (3, -0) {$\pi$};
		\node [style=none] (20) at (3.5, 0.75) {};
		\node [style=gn] (21) at (3.5, -0.25) {};
		\node [style=rn] (22) at (3.5, 0.25) {};
	\end{pgfonlayer}
	\begin{pgfonlayer}{edgelayer}
		\draw (3.center) to (0.center);
		\draw (6) to (5);
		\draw (8.center) to (12);
		\draw (10) to (11.center);
		\draw (16.center) to (13);
		\draw (17) to (15.center);
		\draw (21) to (18.center);
		\draw (20.center) to (22);
	\end{pgfonlayer}
\end{tikzpicture}$
\end{proof}

\subsection{Cyclotomic Supplementarity}

In order to prove the soundness of \suppn, let us first define the equality:
\begin{equation}
\label{eq:supp-no-neighbours}
\begin{tikzpicture}
	\begin{pgfonlayer}{nodelayer}
		\node [style=gn, align=center] (0) at (-0.9999996, -0) {\scriptsize $\alpha{+}$\\ \scriptsize $\frac{n-1}{n}2\pi$};
		\node [style=gn] (1) at (-3.5, -0) {$~\alpha~$};
		\node [style=none] (2) at (-2, -0) {$\cdots$};
		\node [style=gn] (3) at (-2.75, -0) {\scriptsize $\alpha{+}\frac{2\pi}{n}$};
		\node [style=none] (4) at (0, -0) {=};
		\node [style=gn, align=center] (5) at (0.9999996, -0) {\scriptsize $n\alpha$+\\\scriptsize $(n{-}1)\pi$};
	\end{pgfonlayer}
\end{tikzpicture}
\end{equation}
\begin{lemma}
\label{lem:soundness-equivalence}
\suppn is sound $\forall\alpha\in\mathbb{R}$ $\qquad\Leftrightarrow\qquad$ \textnormal{(\ref{eq:supp-no-neighbours})} is sound $\forall\alpha\in\mathbb{R}$.
\end{lemma}

\begin{proof}
\phantomsection\label{prf:soundness-equivalence}
Notice that \begin{tikzpicture}
	\begin{pgfonlayer}{nodelayer}
		\node [style=gn] (0) at (-0.7499999, -0.2500001) {};
		\node [style=none] (1) at (-0.7499999, 0.2500001) {};
		\node [style=none] (2) at (0.2499997, 0.2500001) {};
		\node [style=gn] (3) at (0.2499997, -0.2500001) {$\pi$};
		\node [style=none] (4) at (-0.25, -0.25) {,};
		\node [style=none, anchor=west, xshift=-4pt] (5) at (0.5, -0) {$\left.\vphantom{\rule{1pt}{1.6em}}\right)$};
		\node [style=none, anchor=east, xshift=4pt] (6) at (-1, -0) {$\left(\vphantom{\rule{1pt}{1.6em}}\right.$};
	\end{pgfonlayer}
	\begin{pgfonlayer}{edgelayer}
		\draw (1.center) to (0);
		\draw (2.center) to (3);
	\end{pgfonlayer}
\end{tikzpicture} form a basis of $\mathbb{C}^2$. Hence, using \ko and \bo:\\
$\text{\suppn is sound } \forall\alpha\in\mathbb{R}$\\
$\phantom{.}\qquad\Leftrightarrow \begin{tikzpicture}
	\begin{pgfonlayer}{nodelayer}
		\node [style=gn] (0) at (1.75, -0.7499999) {$k\pi$};
		\node [style=gn] (1) at (-0.2499997, 0.7499999) {\scriptsize $\alpha{+}\frac{n-1}{n}2\pi$};
		\node [style=none] (2) at (0.7499999, -0.2500001) {=};
		\node [style=gn, align=center] (3) at (1.75, 1) {\scriptsize $n\alpha$+\\\scriptsize $(n{-}1)\pi$};
		\node [style=none] (4) at (-1.25, 0.2500001) {$\cdots$};
		\node [style=rn] (5) at (-1.5, -0.2500001) {};
		\node [style=gn] (6) at (-2.75, 0.7499999) {$~\alpha~$};
		\node [style=none] (7) at (1.75, 0.2500001) {$\cdots$};
		\node [style=rn] (8) at (1.75, -0.2500001) {};
		\node [style=gn] (9) at (-2, 0.7499999) {\scriptsize $\alpha{+}\frac{2\pi}{n}$};
		\node [style=gn] (10) at (-1.5, -0.7499999) {$k\pi$};
	\end{pgfonlayer}
	\begin{pgfonlayer}{edgelayer}
		\draw [bend right, looseness=1.00] (6) to (5);
		\draw [bend left, looseness=1.00] (5) to (9);
		\draw [bend left, looseness=1.00] (1) to (5);
		\draw [bend right=75, looseness=1.25] (3) to (8);
		\draw [bend right, looseness=1.00] (3) to (8);
		\draw [bend right=75, looseness=1.25] (8) to (3);
		\draw [style=none] (5) to (10);
		\draw [style=none] (8) to (0);
	\end{pgfonlayer}
\end{tikzpicture} \text{ is sound } \forall\alpha\in\mathbb{R}, \forall k \in\{0,1\}$\\
$\phantom{.}\qquad\Leftrightarrow \begin{tikzpicture}
	\begin{pgfonlayer}{nodelayer}
		\node [style=gn] (0) at (2, -0.9999999) {};
		\node [style=gn, align=center] (1) at (-0.2500001, 0.75) {\scriptsize $\alpha{+}$\\\scriptsize $\frac{n-1}{n}2\pi$};
		\node [style=none] (2) at (0.7500001, -0.5000001) {=};
		\node [style=gn, align=center] (3) at (2, 1) {\scriptsize $n\alpha$+\\\scriptsize $(n{-}1)\pi$};
		\node [style=none] (4) at (-1.25, -0) {$\cdots$};
		\node [style=rn] (5) at (-1.5, -0.5000001) {};
		\node [style=gn] (6) at (-2.75, 0.75) {$~\alpha~$};
		\node [style=none] (7) at (2.25, -0) {$\cdots$};
		\node [style=rn] (8) at (2, -0.5000001) {};
		\node [style=gn] (9) at (-2, 0.75) {\scriptsize $\alpha{+}\frac{2\pi}{n}$};
		\node [style=gn] (10) at (-1.5, -0.9999999) {};
		\node [style=gn] (11) at (-2, -0) {$k\pi$};
		\node [style=gn] (12) at (-2.75, -0) {$k\pi$};
		\node [style=gn] (13) at (-0.7500001, -0.2500001) {$k\pi$};
		\node [style=gn] (14) at (1.25, -0) {$k\pi$};
		\node [style=gn] (15) at (1.75, -0) {$k\pi$};
		\node [style=gn] (16) at (2.75, -0) {$k\pi$};
	\end{pgfonlayer}
	\begin{pgfonlayer}{edgelayer}
		\draw [style=none] (5) to (10);
		\draw [style=none] (8) to (0);
		\draw (6) to (12);
		\draw [bend right=15, looseness=1.00] (12) to (5);
		\draw (9) to (11);
		\draw [bend right, looseness=1.00] (11) to (5);
		\draw [bend left=15, looseness=1.00] (1) to (13);
		\draw [bend left=15, looseness=1.00] (13) to (5);
		\draw [in=90, out=-157, looseness=1.00] (3) to (14);
		\draw [bend right, looseness=1.00] (14) to (8);
		\draw (15) to (3);
		\draw [bend right=15, looseness=1.00] (15) to (8);
		\draw [in=90, out=-23, looseness=1.00] (3) to (16);
		\draw [bend left, looseness=1.00] (16) to (8);
	\end{pgfonlayer}
\end{tikzpicture} \text{ is sound } \forall\alpha\in\mathbb{R},\forall k \in\{0,1\}$\\
$\phantom{.}\qquad\Leftrightarrow \begin{tikzpicture}
	\begin{pgfonlayer}{nodelayer}
		\node [style=none] (0) at (0.7499998, -0) {=};
		\node [style=gn, align=center] (1) at (2, -0) {\scriptsize $n(\alpha{+}k\pi)$+\\\scriptsize $(n{-}1)\pi$};
		\node [style=none] (2) at (-1.5, -0) {$\cdots$};
		\node [style=gn, align=center] (3) at (-0.5, -0) {\scriptsize $\alpha{+}k\pi$ \\\scriptsize ${+}\frac{n-1}{n}2\pi$};
		\node [style=gn, align=center] (4) at (-2.25, -0) {\scriptsize $\alpha{+}k\pi$\\\scriptsize ${+}\frac{2\pi}{n}$};
		\node [style=gn, align=center] (5) at (-3.25, -0) {\scriptsize $\alpha{+}k\pi$};
	\end{pgfonlayer}
\end{tikzpicture} \text{ is sound } \forall\alpha\in\mathbb{R},\forall k \in\{0,1\}$\\
$\phantom{.}\qquad\Leftrightarrow \text{(\ref{eq:supp-no-neighbours}) is sound } \forall\alpha\in\mathbb{R}$
\end{proof}

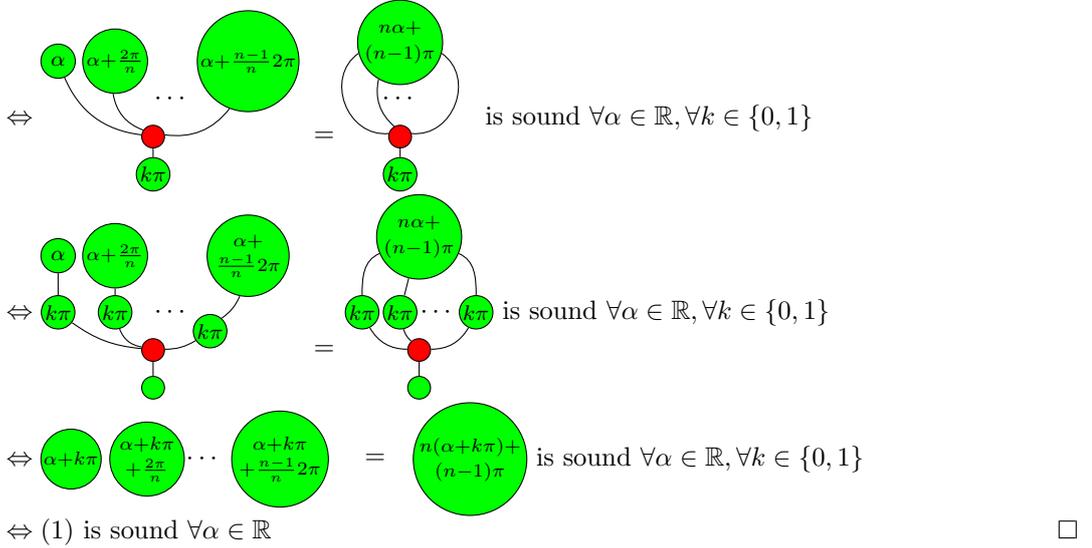
\begin{proof}[Proof of Proposition \ref{prop:suppn-sound}]
\phantomsection\label{prf:suppn-sound}
According to the previous lemma, \suppn is sound $\forall\alpha\in\mathbb{R}$ if and only if (\ref{eq:supp-no-neighbours}) is sound $\forall\alpha\in\mathbb{R}$. The standard interpretation of the right part of the equality \ref{eq:supp-no-neighbours} yields $1+e^{i(n\alpha+(n-1)\pi)}$, and, using the fact that $\prod\limits_{k=0}^n (X-e^{\frac{2ik\pi}{n}}) = X^n-1$ (the $e^{\frac{2ik\pi}{n}}$ are the $n^{th}$ roots of unity), the interpretation of the left part amounts to:
\begin{align*}
\prod_{k=0}^{n-1} (1+e^{i(\alpha+\frac{2k\pi}{n})}) &= e^{in(\alpha+\pi)} \prod_{k=0}^{n-1} (e^{-i(\alpha+\pi)}-e^{\frac{2ik\pi}{n}}) = e^{in(\alpha+\pi)}(e^{-in(\alpha+\pi)}-1) \\
&= 1+e^{i(n\alpha+(n-1)\pi)}
\end{align*}
Hence (\ref{eq:supp-no-neighbours}) is sound which implies \suppn is sound.
\end{proof}

\begin{proof}[Proof of Proposition \ref{prop:cyclotomic-neighbourhood}]
\phantomsection\label{prf:cyclotomic-neighbourhood}
Consider $n$ green cyclotomic twins for $n\in\mathbb{N}^*$. Notice that since the dots are twins, they are either an independent set or a clique – if two of them are connected, they are all connected. If another green dot is part of the neighbourhood, or if the twins are connected, they can be merged using \sone. Otherwise, up to \hlaw, the neighbourhood is composed of only $m$ red dots connected with exactly $1$ wire to any of the twins. If $m=0$, the equation becomes (\ref{eq:supp-no-neighbours}), proven in lemma \ref{lem:soundness-equivalence}. If $m=1$, then the direct use of \suppn gives the result. If $m\geq 1$, we have to use the generalised bialgebra rule proven in \cite{euler-decomp}:
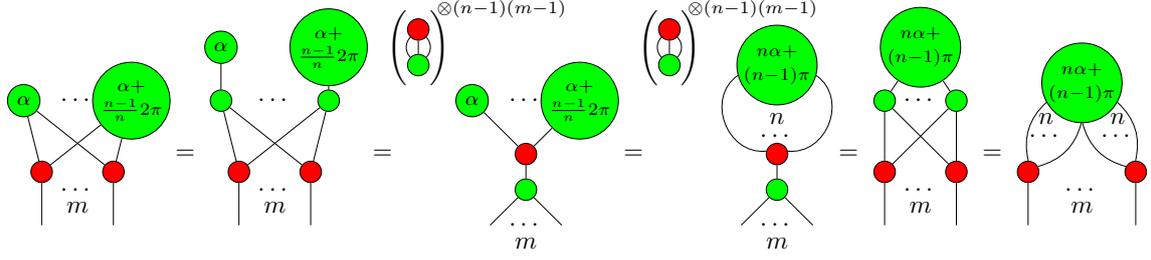
\[\fit{\begin{tikzpicture}
	\begin{pgfonlayer}{nodelayer}
		\node [style=none] (0) at (-5, 0.75) {$\cdots$};
		\node [style=gn] (1) at (-5.75, 0.75) {};
		\node [style=rn] (2) at (-5.5, -0.25) {};
		\node [style=gn, align=center] (3) at (-4.25, 0.75) {};
		\node [style=none] (4) at (-5.5, -1) {};
		\node [style=rn] (5) at (-4.5, -0.25) {};
		\node [style=none] (6) at (-4.5, -1) {};
		\node [style=none] (7) at (-3.5, -0) {=};
		\node [style=none] (8) at (-5, -0.5) {$\cdots$};
		\node [style=gn] (9) at (-2.25, 0.75) {$~\alpha~$};
		\node [style=none] (10) at (-2, -1) {};
		\node [style=gn, align=center] (11) at (-0.75, 0.75) {\scriptsize $\alpha{+}$\\ \scriptsize $\frac{n-1}{n}2\pi$};
		\node [style=rn] (12) at (-1.5, -0) {};
		\node [style=none] (13) at (-1.5, -1) {$\cdots$};
		\node [style=none] (14) at (-1, -1) {};
		\node [style=gn] (15) at (-1.5, -0.5000002) {};
		\node [style=none] (16) at (-5, -0.75) {$m$};
		\node [style=none] (17) at (-1.5, -1.25) {$m$};
		\node [style=none] (18) at (0, -0) {=};
		\node [style=none] (19) at (1.5, -1) {};
		\node [style=gn] (20) at (2, -0.5) {};
		\node [style=none] (21) at (2, -1.25) {$m$};
		\node [style=none] (22) at (2.5, -1) {};
		\node [style=none] (23) at (2, -1) {$\cdots$};
		\node [style=none] (24) at (2, 0.25) {$\cdots$};
		\node [style=gn, align=center] (25) at (2, 1.25) {\scriptsize $n\alpha$+\\\scriptsize $(n{-}1)\pi$};
		\node [style=rn] (26) at (2, -0) {};
		\node [style=none] (27) at (5, -0) {=};
		\node [style=none] (28) at (6.25, -0.5) {$\cdots$};
		\node [style=rn] (29) at (5.5, -0.25) {};
		\node [style=none] (30) at (5.5, -1) {};
		\node [style=none] (31) at (6.25, -0.75) {$m$};
		\node [style=gn, align=center] (32) at (6.25, 1) {\scriptsize $n\alpha$+\\\scriptsize $(n{-}1)\pi$};
		\node [style=none] (33) at (7, -1) {};
		\node [style=rn] (34) at (7, -0.25) {};
		\node [style=none] (35) at (5.75, 0.5) {$n$};
		\node [style=none] (36) at (5.75, 0.25) {$\cdots$};
		\node [style=none] (37) at (6.75, 0.5) {$n$};
		\node [style=none] (38) at (6.75, 0.25) {$\cdots$};
		\node [style=none] (39) at (-1.5, 0.75) {$\cdots$};
		\node [style=none] (40) at (2, 0.5) {$n$};
		\node [style=none] (41) at (-6.25, -0) {=};
		\node [style=none] (42) at (-7.75, 0.75) {$\cdots$};
		\node [style=none] (43) at (-8.25, -1) {};
		\node [style=none] (44) at (-7.25, -1) {};
		\node [style=rn] (45) at (-7.25, -0.25) {};
		\node [style=gn] (46) at (-8.5, 0.75) {$~\alpha~$};
		\node [style=none] (47) at (-7.75, -0.75) {$m$};
		\node [align=center, style=gn] (48) at (-7, 0.75) {\scriptsize $\alpha{+}$\\ \scriptsize $\frac{n-1}{n}2\pi$};
		\node [style=rn] (49) at (-8.25, -0.25) {};
		\node [style=none] (50) at (-7.75, -0.5) {$\cdots$};
		\node [style=gn] (51) at (-5.75, 1.5) {$~\alpha~$};
		\node [align=center, style=gn] (52) at (-4.25, 1.5) {\scriptsize $\alpha{+}$\\ \scriptsize $\frac{n-1}{n}2\pi$};
		\node [style=none] (53) at (3, -0) {=};
		\node [style=none] (54) at (4, 0.75) {$\cdots$};
		\node [style=none] (55) at (3.5, -1) {};
		\node [style=gn] (56) at (3.5, 0.75) {};
		\node [style=rn] (57) at (4.5, -0.25) {};
		\node [style=none] (58) at (4.5, -1) {};
		\node [style=none] (59) at (4, -0.75) {$m$};
		\node [align=center, style=gn] (60) at (4.5, 0.75) {};
		\node [style=rn] (61) at (3.5, -0.25) {};
		\node [style=none] (62) at (4, -0.5) {$\cdots$};
		\node [align=center, style=gn] (63) at (4, 1.5) {\scriptsize $n\alpha$+\\\scriptsize $(n{-}1)\pi$};
		\node [style=none, anchor=west, xshift=-4pt] (64) at (-2.75, 1.5) {$\left.\vphantom{\rule{1pt}{1.6em}}\right)$};
		\node [style=none, anchor=east, xshift=4pt] (65) at (-3.25, 1.5) {$\left(\vphantom{\rule{1pt}{1.6em}}\right.$};
		\node [style=gn] (66) at (-3, 1.25) {};
		\node [style=none, anchor=west] (67) at (-2.75, 2) {$\vphantom{.}^{\otimes (n{-}1)(m{-}1)}$};
		\node [style=rn] (68) at (-3, 1.75) {};
		\node [xshift=4pt, anchor=east, style=none] (69) at (0.25, 1.5) {$\left(\vphantom{\rule{1pt}{1.6em}}\right.$};
		\node [xshift=-4pt, anchor=west, style=none] (70) at (0.75, 1.5) {$\left.\vphantom{\rule{1pt}{1.6em}}\right)$};
		\node [anchor=west, style=none] (71) at (0.75, 2) {$\vphantom{.}^{\otimes (n{-}1)(m{-}1)}$};
		\node [style=rn] (72) at (0.5, 1.75) {};
		\node [style=gn] (73) at (0.5, 1.25) {};
	\end{pgfonlayer}
	\begin{pgfonlayer}{edgelayer}
		\draw (1) to (2);
		\draw (3) to (2);
		\draw (2) to (4.center);
		\draw (5) to (6.center);
		\draw (1) to (5);
		\draw (3) to (5);
		\draw (9) to (12);
		\draw (11) to (12);
		\draw (12) to (15);
		\draw (15) to (10.center);
		\draw (15) to (14.center);
		\draw (20) to (19.center);
		\draw (20) to (22.center);
		\draw [bend right=75, looseness=1.25] (25) to (26);
		\draw [bend right=75, looseness=1.25] (26) to (25);
		\draw (26) to (20);
		\draw [bend right, looseness=1.00] (32) to (29);
		\draw [bend right, looseness=1.00] (29) to (32);
		\draw (29) to (30.center);
		\draw (34) to (33.center);
		\draw [bend right, looseness=1.00] (32) to (34);
		\draw [bend right, looseness=1.00] (34) to (32);
		\draw (46) to (49);
		\draw (48) to (49);
		\draw (49) to (43.center);
		\draw (45) to (44.center);
		\draw (46) to (45);
		\draw (48) to (45);
		\draw (52) to (3);
		\draw (51) to (1);
		\draw (56) to (61);
		\draw (60) to (61);
		\draw (61) to (55.center);
		\draw (57) to (58.center);
		\draw (56) to (57);
		\draw (60) to (57);
		\draw (63) to (56);
		\draw (63) to (60);
		\draw [bend right=45, looseness=1.00] (68) to (66);
		\draw [bend right=45, looseness=1.00] (66) to (68);
		\draw (68) to (66);
		\draw [bend right=45, looseness=1.00] (72) to (73);
		\draw [bend right=45, looseness=1.00] (73) to (72);
		\draw (72) to (73);
	\end{pgfonlayer}
\end{tikzpicture}}\]
The proof is the same for red cyclotomic twins.
\end{proof}

\subsection{Derivations of the Supplementarity}

\begin{proof}[Proof of Proposition \ref{prop:sp-spq-to-sq}]
\phantomsection\label{prf:sp-spq-to-sq}
With \hlaw the Hopf Law and \iv the inverse rule, which are both derivable from $ZX_E$:
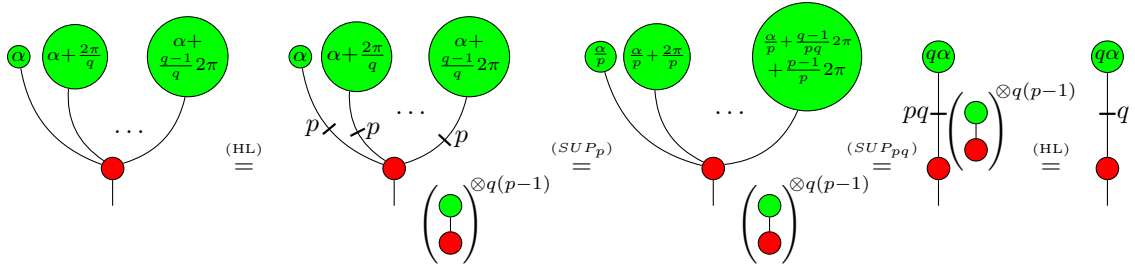
\[\fit{\begin{tikzpicture}
	\begin{pgfonlayer}{nodelayer}
		\node [style=none] (0) at (-6, -0.75) {};
		\node [style=rn] (1) at (-6, -0.25) {};
		\node [style=gn] (2) at (-7.25, 1.25) {$\alpha$};
		\node [style=gn] (3) at (-6.5, 1.25) {\scriptsize $\alpha{+}\frac{2\pi}{q}$};
		\node [style=gn, align=center] (4) at (-5, 1.25) {\scriptsize $\alpha{+}$\\ \scriptsize $\frac{q-1}{q}2\pi$};
		\node [style=gn] (5) at (-1.5, -0.7500001) {};
		\node [style=rn] (6) at (-1.5, -1.25) {};
		\node [style=none] (7) at (-5.75, 0.25) {$\dots$};
		\node [style=none] (8) at (-4.25, -0.25) {=};
		\node [style=gn, align=center] (9) at (-1.25, 1.25) {\scriptsize $\alpha{+}$\\ \scriptsize$\frac{q-1}{q}2\pi$};
		\node [style=rn] (10) at (-2.25, -0.25) {};
		\node [style=gn] (11) at (-3.5, 1.25) {$\alpha$};
		\node [style=none] (12) at (-2.25, -0.75) {};
		\node [style=none] (13) at (-2, 0.5) {$\dots$};
		\node [style=gn] (14) at (-2.75, 1.25) {$\alpha{+}\frac{2\pi}{q}$};
		\node [style=gn] (19) at (1.25, 1.25) {\tiny $\frac{\alpha}{p}{+}\frac{2\pi}{p}$};
		\node [style=rn] (20) at (2, -0.25) {};
		\node [style=none] (21) at (2.25, 0.5) {$\dots$};
		\node [style=none] (22) at (2, -0.75) {};
		\node [style=gn, align=center] (23) at (3.25, 1.25) {\tiny $\frac{\alpha}{p}{+}\frac{q-1}{pq}2\pi$\\ \scriptsize $+ \frac{p-1}{p}2\pi$};
		\node [style=gn] (24) at (0.5000001, 1.25) {\tiny $\frac{\alpha}{p}$};
		\node [style=none] (25) at (0.25, -0.25) {=};
		\node [style=none] (26) at (-4.25, -0) {\tiny (HL)};
		\node [style=none] (27) at (0.25, -0) {\tiny $(SUP_p)$};
		\node [style=rn] (28) at (2.75, -1.25) {};
		\node [style=gn] (29) at (2.75, -0.7500001) {};
		\node [style=none] (31) at (5, -0.75) {};
		\node [style=rn] (32) at (5.5, -0) {};
		\node [style=rn] (33) at (5, -0.25) {};
		\node [style=none] (35) at (4.25, -0.25) {=};
		\node [style=none] (36) at (4.25, -0) {\tiny $(SUP_{pq})$};
		\node [style=gn,align=center] (37) at (5, 1.25) {$q\alpha$};
		\node [style=gn] (38) at (5.5, 0.5) {};
		\node [style=rn] (40) at (7.250001, -0.25) {};
		\node [style=none] (42) at (6.5, -0) {\tiny (HL)};
		\node [style=none] (43) at (6.5, -0.25) {=};
		\node [style=none] (44) at (7.250001, -0.75) {};
		\node [style=gn,align=center] (45) at (7.250001, 1.25) {$q\alpha$};
		\node [style=none, anchor=west, xshift=-4pt] (138) at (-1.25, -1) {$\left.\vphantom{\rule{1pt}{1.6em}}\right)$};
		\node [style=none, anchor=east, xshift=4pt] (139) at (-1.75, -1) {$\left(\vphantom{\rule{1pt}{1.6em}}\right.$};
		\node [style=none, anchor=west] (140) at (-1.25, -0.5) {$\vphantom{.}^{\otimes q(p{-}1)}$};
		\node [xshift=-4pt, anchor=west, style=none] (141) at (3, -1) {$\left.\vphantom{\rule{1pt}{1.6em}}\right)$};
		\node [anchor=west, style=none] (142) at (3, -0.5) {$\vphantom{.}^{\otimes q(p{-}1)}$};
		\node [xshift=4pt, anchor=east, style=none] (143) at (2.5, -1) {$\left(\vphantom{\rule{1pt}{1.6em}}\right.$};
		\node [xshift=-4pt, anchor=west, style=none] (144) at (5.75, 0.25) {$\left.\vphantom{\rule{1pt}{1.6em}}\right)$};
		\node [anchor=west, style=none] (145) at (5.75, 0.75) {$\vphantom{.}^{\otimes q(p{-}1)}$};
		\node [xshift=4pt, anchor=east, style=none] (146) at (5.25, 0.25) {$\left(\vphantom{\rule{1pt}{1.6em}}\right.$};
	\end{pgfonlayer}
	\begin{pgfonlayer}{edgelayer}
		\draw [style=none, bend right, looseness=1.00] (2) to (1);
		\draw [style=none] (1) to (0.center);
		\draw [style=none, bend right, looseness=1.00] (3) to (1);
		\draw [style=none, bend right, looseness=1.00] (1) to (4);
		\draw [style=none] (5) to (6);
		\draw [style=none] (10) to (12.center);
		\path  (10) edge[style=tickedge, bend left] node[left] {$p$} (11);
		\path  (10) edge[style=tickedge, bend left] node[right] {$p$} (14);
		\path  (10) edge[style=tickedge, bend right] node[right] {$p$}  (9);
		\draw [style=none, bend right, looseness=1.00] (24) to (20);
		\draw [style=none] (20) to (22.center);
		\draw [style=none, bend right, looseness=1.00] (19) to (20);
		\draw [style=none, bend right, looseness=1.00] (20) to (23);
		\draw [style=none] (29) to (28);
		\draw [style=none] (33) to (31.center);
		\draw [style=none] (38) to (32);
		\path (37) edge[style=tickedge] node[left] {$pq$} (33);
		\draw [style=none] (40) to (44.center);
		\path (45) edge[style=tickedge] node[right] {$q$} (40);
	\end{pgfonlayer}
\end{tikzpicture}
}\]
With $p$-ticked edge representing $p$ parallel wires. Those are created -- always by multiples of two -- using the Hopf Law \hlaw. The $3^{rd}$ diagram is obtained by rearranging the branches so that we can use \suppc{pq}. For the last two diagrams, notice that $q$ is odd so $(q-1)\pi=0\mod 2\pi$.
\end{proof}

\begin{proof}[Proof of Proposition \ref{prop:sp-sp2q-to-spq}]
\phantomsection\label{prf:sp-sp2q-to-spq}
If $p$ is odd, then the previous proposition implies the wanted result, taking $q:=pq$.\\
Now, if $p=0\mod 2$:
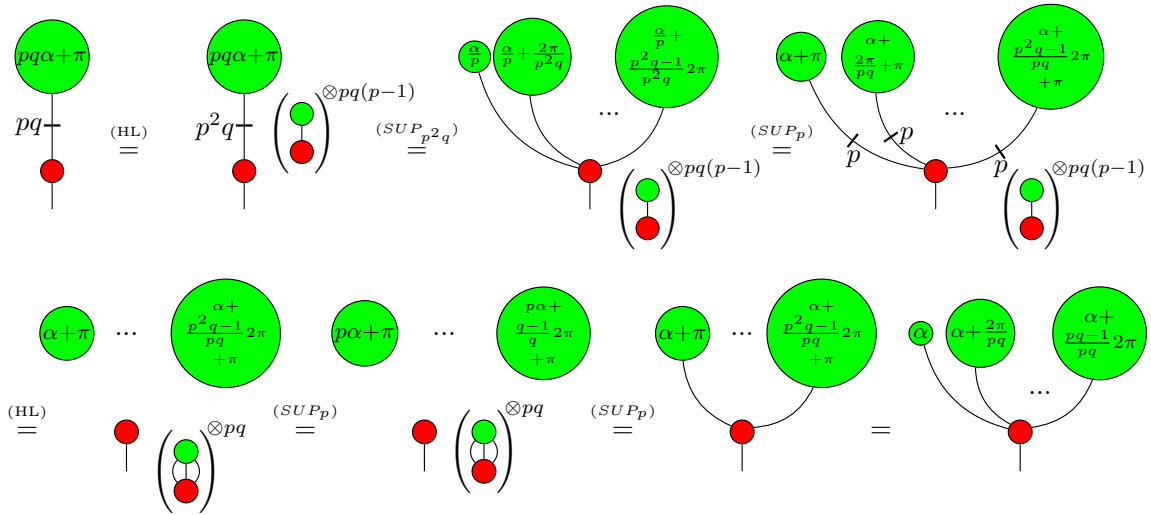
\[\fit{\begin{tikzpicture}
	\begin{pgfonlayer}{nodelayer}
		\node [align=center, style=gn] (0) at (6.500001, 1.25) {\tiny $\alpha{+}$\\ \tiny $\frac{p^2q-1}{pq}2\pi$\\ \tiny ${+}\pi$};
		\node [style=rn] (1) at (5.000001, -0.25) {};
		\node [style=gn] (2) at (3.25, 1.25) {\scriptsize $\alpha{+}\pi$};
		\node [style=none] (3) at (5.000001, -0.75) {};
		\node [style=none] (4) at (5.250001, 0.5) {...};
		\node [align=center, style=gn] (5) at (4.250001, 1.25) {\tiny $\alpha{+}$\\ \tiny $\frac{2\pi}{pq}{+}\pi$};
		\node [style=gn] (6) at (-0.25, 1.25) {\tiny $\frac{\alpha}{p}{+}\frac{2\pi}{p^2q}$};
		\node [style=rn] (7) at (0.5, -0.25) {};
		\node [style=none] (8) at (0.75, 0.5) {...};
		\node [style=none] (9) at (0.5, -0.75) {};
		\node [align=center, style=gn] (10) at (1.5, 1.25) {\tiny $\frac{\alpha}{p}{+}$\\ \tiny $\frac{p^2q-1}{p^2q}2\pi$};
		\node [style=gn] (11) at (-1, 1.25) {\tiny $\frac{\alpha}{p}$};
		\node [style=none] (12) at (3, -0) {=};
		\node [style=none] (13) at (3, 0.25) {\tiny $(SUP_p)$};
		\node [style=none] (14) at (-4, -0.75) {};
		\node [style=rn] (15) at (-3.25, -0) {};
		\node [style=rn] (16) at (-4, -0.25) {};
		\node [style=none] (18) at (-1.75, -0) {=};
		\node [style=none] (19) at (-1.75, 0.25) {\tiny $(SUP_{p^2q})$};
		\node [style=gn] (20) at (-4, 1.25) {$pq\alpha{+}\pi$};
		\node [style=gn] (21) at (-3.25, 0.5) {};
		\node [style=rn] (22) at (-6.5, -0.25) {};
		\node [style=none] (23) at (-5.5, 0.25) {\tiny (HL)};
		\node [style=none] (24) at (-5.5, -0) {=};
		\node [style=none] (25) at (-6.5, -0.75) {};
		\node [style=gn] (26) at (-6.5, 1.25) {$pq\alpha{+}\pi$};
		\node [style=rn] (27) at (1.25, -1) {};
		\node [style=gn] (28) at (1.25, -0.5) {};
		\node [style=rn] (31) at (6.250001, -1) {};
		\node [style=gn] (32) at (6.250001, -0.5) {};
		\node [style=none, anchor=west, xshift=-4pt] (45) at (-3, 0.25) {$\left.\vphantom{\rule{1pt}{1.6em}}\right)$};
		\node [style=none, anchor=east, xshift=4pt] (46) at (-3.5, 0.25) {$\left(\vphantom{\rule{1pt}{1.6em}}\right.$};
		\node [style=none, anchor=west] (47) at (-3, 0.75) {$\vphantom{.}^{\otimes pq(p{-}1)}$};
		\node [xshift=4pt, anchor=east, style=none] (48) at (1, -0.75) {$\left(\vphantom{\rule{1pt}{1.6em}}\right.$};
		\node [xshift=-4pt, anchor=west, style=none] (49) at (1.5, -0.75) {$\left.\vphantom{\rule{1pt}{1.6em}}\right)$};
		\node [anchor=west, style=none] (50) at (1.5, -0.25) {$\vphantom{.}^{\otimes pq(p{-}1)}$};
		\node [xshift=4pt, anchor=east, style=none] (51) at (6, -0.75) {$\left(\vphantom{\rule{1pt}{1.6em}}\right.$};
		\node [xshift=-4pt, anchor=west, style=none] (52) at (6.5, -0.75) {$\left.\vphantom{\rule{1pt}{1.6em}}\right)$};
		\node [anchor=west, style=none] (53) at (6.5, -0.25) {$\vphantom{.}^{\otimes pq(p{-}1)}$};
	\end{pgfonlayer}
	\begin{pgfonlayer}{edgelayer}
		\draw [style=none] (1) to (3.center);
		\draw [style=none, bend right, looseness=1.00] (11) to (7);
		\draw [style=none] (7) to (9.center);
		\draw [style=none, bend right, looseness=1.00] (6) to (7);
		\draw [style=none, bend right, looseness=1.00] (7) to (10);
		\draw [style=none] (16) to (14.center);
		\draw [style=none] (21) to (15);
		\draw [style=none] (22) to (25.center);
		\draw [style=none] (28) to (27);
		\draw [style=none] (32) to (31);
		\path (2) edge[style=tickedge, bend right, looseness=1.00] node[below] {$p$} (1);
		\path (1) edge[style=tickedge, bend left, looseness=1.00] node[right] {$p$} (5);
		\path (1) edge[style=tickedge, bend right, looseness=1.00] node[below] {$p$} (0);
		\path (26) edge[style=tickedge] node[left] {$pq$} (22);
		\path (20) edge[style=tickedge] node[left] {$p^2q$} (16);
	\end{pgfonlayer}
\end{tikzpicture}}\]
\[\fit{\begin{tikzpicture}
	\begin{pgfonlayer}{nodelayer}
		\node [style=none] (0) at (5.750001, -0.5) {};
		\node [style=rn] (1) at (5.750001, -0) {};
		\node [style=gn] (2) at (4.500001, 1.25) {$\alpha$};
		\node [style=gn] (3) at (5.250001, 1.25) {\scriptsize $\alpha{+}\frac{2\pi}{pq}$};
		\node [align=center, style=gn] (4) at (6.750001, 1.25) {\scriptsize $\alpha{+}$\\ \scriptsize $\frac{pq-1}{pq}2\pi$};
		\node [style=none] (5) at (6.000001, 0.5) {...};
		\node [style=none] (6) at (-6.75, -0) {=};
		\node [style=none] (7) at (-6.75, 0.25) {\tiny (HL)};
		\node [align=center, style=gn] (8) at (-4.25, 1.25) {\tiny $\alpha{+}$\\ \tiny $\frac{p^2q-1}{pq}2\pi$\\ \tiny ${+}\pi$};
		\node [style=gn] (9) at (-4.75, -0.25) {};
		\node [style=rn] (10) at (-5.5, -0) {};
		\node [style=rn] (11) at (-4.75, -0.75) {};
		\node [style=gn] (12) at (-6.25, 1.25) {$\alpha{+}\pi$};
		\node [style=none] (13) at (-5.5, -0.5) {};
		\node [style=none] (15) at (-5.5, 1.25) {...};
		\node [style=none] (16) at (-3.25, 0.25) {\tiny $(SUP_p)$};
		\node [style=none] (17) at (-3.25, -0) {=};
		\node [style=gn] (18) at (-1, -0) {};
		\node [style=rn] (19) at (-1.75, -0) {};
		\node [style=none] (20) at (-1.75, -0.5) {};
		\node [style=none] (21) at (-1.5, 1.25) {...};
		\node [style=rn] (22) at (-1, -0.5) {};
		\node [style=gn] (24) at (-2.5, 1.25) {$p\alpha{+}\pi$};
		\node [align=center, style=gn] (25) at (-0.25, 1.25) {\tiny $p\alpha+$\\ \tiny $\frac{q-1}{q}2\pi$\\ \tiny $+\pi$};
		\node [style=none] (26) at (0.75, -0) {=};
		\node [style=none] (27) at (0.75, 0.25) {\tiny $(SUP_p)$};
		\node [style=rn] (28) at (2.25, -0) {};
		\node [style=none] (29) at (2.25, -0.5) {};
		\node [style=none] (30) at (2.25, 1.25) {...};
		\node [style=gn] (31) at (1.5, 1.25) {$\alpha{+}\pi$};
		\node [align=center, style=gn] (32) at (3.25, 1.25) {\tiny $\alpha{+}$ \\ \tiny $\frac{p^2q-1}{pq}2\pi$ \\ \tiny ${+}\pi$};
		\node [style=none] (33) at (4.000001, -0) {=};
		\node [style=none, anchor=west, xshift=-4pt] (42) at (-0.75, -0.25) {$\left.\vphantom{\rule{1pt}{1.6em}}\right)$};
		\node [style=none, anchor=east, xshift=4pt] (43) at (-1.25, -0.25) {$\left(\vphantom{\rule{1pt}{1.6em}}\right.$};
		\node [style=none, anchor=west] (44) at (-0.75, 0.25) {$\vphantom{.}^{\otimes pq}$};
		\node [anchor=west, style=none] (45) at (-4.5, -0) {$\vphantom{.}^{\otimes pq}$};
		\node [xshift=-4pt, anchor=west, style=none] (46) at (-4.5, -0.5) {$\left.\vphantom{\rule{1pt}{1.6em}}\right)$};
		\node [xshift=4pt, anchor=east, style=none] (47) at (-5, -0.5) {$\left(\vphantom{\rule{1pt}{1.6em}}\right.$};
	\end{pgfonlayer}
	\begin{pgfonlayer}{edgelayer}
		\draw [style=none, bend right, looseness=1.00] (2) to (1);
		\draw [style=none] (1) to (0.center);
		\draw [style=none, bend right, looseness=1.00] (3) to (1);
		\draw [style=none, bend right, looseness=1.00] (1) to (4);
		\draw [style=none] (10) to (13.center);
		\draw [style=none, bend right=45, looseness=1.00] (9) to (11);
		\draw [style=none, bend left=45, looseness=1.00] (9) to (11);
		\draw [style=none] (11) to (9);
		\draw [style=none] (19) to (20.center);
		\draw [style=none, bend right=45, looseness=1.00] (18) to (22);
		\draw [style=none, bend left=45, looseness=1.00] (18) to (22);
		\draw [style=none] (22) to (18);
		\draw [style=none] (28) to (29.center);
		\draw [style=none, bend right, looseness=1.00] (31) to (28);
		\draw [style=none, bend right, looseness=1.00] (28) to (32);
	\end{pgfonlayer}
\end{tikzpicture}}\]
The second call to \suppc{p} makes use of the equality (\ref{eq:supp-no-neighbours}) of lemma \ref{lem:soundness-equivalence}. The last equality is just a rearranging of the branches.
\end{proof}

\end{document}